\documentclass[10pt,a4paper]{amsart}
\usepackage[margin=2.5cm]{geometry}
\usepackage[utf8]{inputenc}
\usepackage{amsthm}
\usepackage{amssymb}
\usepackage{amsmath}
\usepackage{mathtools}
\usepackage{xcolor}
\usepackage{mathabx}
\usepackage{stmaryrd}
\usepackage{hyperref}
\hypersetup{colorlinks=true,linkcolor=blue,citecolor=teal,filecolor=magenta,urlcolor=cyan}

\newcommand{\lb}{\left (}
\newcommand{\rb}{\right )}
\newcommand{\pb}{\bar{p}}
\newcommand{\pt}{\tilde{p}}
\newcommand{\ext}{\mathrm{ext}}

\newcommand{\Fc}{\mathcal{F}}

\newcommand{\cV}{\mathcal{V}}
\newcommand{\Z}{\mathbb{Z}}

\newcommand{\mrq}{\mathrm{q}}
\newcommand{\Dop}{\mathcal{D}}
\DeclareMathOperator{\Res}{Res}

\newtheorem{theorem}{Theorem}[section]
\newtheorem{proposition}[theorem]{Proposition}
\newtheorem{corollary}[theorem]{Corollary}
\newtheorem{lemma}[theorem]{Lemma}
\theoremstyle{definition}
\newtheorem{notation}[theorem]{Notation}
\newtheorem{remark}[theorem]{Remark}
\newtheorem{definition}[theorem]{Definition}
\newtheorem{convention}[theorem]{Convention}

\numberwithin{equation}{section}

\newcommand{\p}[0]{\partial}

\def\bigcovac{\big\langle 0\, \big |}
\def\bigvac{\big | \, 0 \big \rangle}

\def\vac{ | 0 \rangle}

\def\cF{\mathcal{F}}
\def\cQ{\mathcal{Q}}
\def\cS{\mathcal{S}}
\def\cW{\mathcal{W}}
\def\halfZ{\mathbb{Z}+{\textstyle \frac 12 }}

\title[Topological recursion for the extended Ooguri--Vafa partition function]{Topological recursion for the extended Ooguri--Vafa partition function of colored HOMFLY-PT polynomials of torus knots}

\author[P.~Dunin-Barkowski]{Petr~Dunin-Barkowski}
\address[P.~Dunin-Barkowski]{Faculty of Mathematics, National Research University Higher School of Economics, Usacheva 6, 119048 Moscow, Russia; HSE--Skoltech International Laboratory of Representation Theory and Mathematical Physics, Skoltech, Nobelya 1, 143026, Moscow, Russia; and ITEP, 117218 Moscow, Russia}
\email{ptdunin@hse.ru}

\author[M.~Kazarian]{Maxim~Kazarian}
\address[M.~Kazarian]{Faculty of Mathematics, National Research University Higher School of Economics, Usacheva 6, 119048 Moscow, Russia; and Center for Advanced Studies, Skoltech, Nobelya 1, 143026, Moscow, Russia}
\email{kazarian@mccme.ru}

\author[A.~Popolitov]{Aleksandr~Popolitov}
\address[A.~Popolitov]{Institute for Information Transmission Problems, Moscow 127994, Russia; and ITEP, Moscow 117218, Russia; and Moscow Institute of Physics and Technology, Dolgoprudny 141701, Russia}
\email{popolit@gmail.com}

\author[S.~Shadrin]{Sergey~Shadrin}
\address[S.~Shadrin]{Korteweg-de Vries Institute for Mathematics, University of Amsterdam, Postbus 94248, 1090 GE Amsterdam, The Netherlands}
\email{s.shadrin@uva.nl}	

\author[A.~Sleptsov]{Alexey~Sleptsov}
\address[A.~Sleptsov]{ITEP, Moscow 117218, Russia; Institute for Information Transmission Problems,	Moscow 127994, Russia; and Moscow Institute of Physics and Technology, Dolgoprudny 141701, Russia}
\email{sleptsov@itep.ru}

\begin{document}

\begin{abstract} We prove that 
	topological recursion applied to the spectral curve of colored HOMFLY-PT polynomials of torus knots reproduces the $n$-point functions of a particular partition function called the \emph{extended Ooguri--Vafa partition function}. This generalizes and refines the results of Brini--Eynard--Mari\~no and Borot--Eynard--Orantin. 
	
We also discuss how the statement of spectral curve topological recursion in this case fits into the program of Alexandrov--Chapuy--Eynard--Harnad of establishing the topological recursion for general weighted double Hurwitz numbers 
partition functions (a.k.a. KP tau-functions of hypergeometric type). 
\end{abstract}

\maketitle

\tableofcontents

\section{Introduction}

The spectral curve 
topological recursion for colored $SU(N)$ torus knot invariants was first proposed by Brini--Eynard--Mari\~{n}o in~\cite{BEM11}. In that paper they investigated a conjectural equivalence (mirror symmetry) between A and B models of open topological string theory. They proposed a B-model geometry, which generalized the previously obtained results for a colored framed unknot and made it possible to do explicit calculations in this geometry. The proposed procedure assigns the data of a spectral curve to each torus knot $T[Q,P]$ and allows one to calculate its $SU(N)$ invariants using topological recursion. The topological recursion was derived from the loop equations of the matrix model that precisely describes these knot invariants, originally proposed in \cite{lawrence1999witten, dolivet2007chern, marino2005chern}. Although this derivation of topological recursion in~\cite{BEM11} should rather be considered a conjecture, soon the development of the matrix model technique allowed Borot--Eynard--Orantin to prove the statement of Brini--Eynard--Mari\~no, see \cite{BEO} (see also~\cite{BorotEynardWeisse,BorotGuionnetKozlowski} for clarification of some further technical details).

To be more precise, we recall that the spectral curve topological recursion (originally studied in the matrix model context for partular matrix models in e.g. \cite{AMM-1,AMM-2,CE1,CE2,CEO} and then formally defined for arbitrary abstract spectral curves in~\cite{EynardOrantin,OrantinThesis}, see also \cite{EynardSurvey}) is a recursive procedure that in its simplest version associates to a Riemann surface $\Sigma$ with some extra data (the so-called spectral curve initial data) symmetric meromorphic $n$-differentials $\omega_{g,n}$ defined on $\Sigma^n$. The conjecture of Brini--Eynard--Mari\~{n}o / theorem of Borot--Eynard--Orantin states that these differentials contain the information about the Ooguri--Vafa partition function $Z$ for the torus knot $T[Q,P]$: 
\begin{equation}
Z=Z(\mrq,A\, |\, \bar p) \coloneqq \sum_R  H_R(T[Q,P];\mrq,A) \, s_R(\pb)
\end{equation}
(we explain the ingredients of this formula and recall all necessary definitions below). This partition function is the generating function of the colored $SU(N)$ invariants (related to the colored HOMFLY polynomials $H_R(T[Q,P];\mrq,A)$ by specializing the variable $A = \mrq^N$) of the knot $T[Q,P]$ and describes the expansion of the Chern-Simons theory around the trivial flat connection.

This statement, however, leaves a feeling of dissatisfaction, since one has to cut off plenty of information contained in the differentials $\omega_{g,n}$ via some symmetrization procedure to get the coefficients of the Ooguri--Vafa partition function, and it is by no means a natural thing to do from the point of view of the topological recursion theory. On the other hand, it was observed in~\cite{DPSS} that the Ooguri--Vafa function has the same drawback from the point of view of the theory of KP integrability. Namely, following the ideas of~\cite{mironov2013character,MMS13} it was shown in~\cite{DPSS} that there exists a very natural extension of the Ooguri--Vafa parition function to a tau-function of KP hierarchy of hypergeometric type~\cite{KMMM95,OrlovScherbin}, more precisely, to a reparametrized generating function of double Hurwitz numbers. Note that in order to get the Ooguri--Vafa parition function from the extended Ooguri--Vafa parition function one has to cut off plenty of coefficients via exactly the same symmetrization procedure as the one applied to $\omega_{g,n}$ in the topological recursion setup. 

These considerations led to a natural conjecture, posed in~\cite{DPSS}, that related full expansions of the correlation differentials $\omega_{g,n}$ constructed via topological recursion from the Brini--Eynard--Mari\~no spectral curve data with the extended Ooguri--Vafa partition functions. A thorough combinatorial analysis of the latter function performed in~\cite{DPSS} confirmed many expected properties that would follow from this conjecture. In fact, the results of~\cite{DPSS} reduced this conjecture to the so-called quadratic loop equation~\cite{BEO,BS}, and its proof (and, as an immediate corollary, a proof of the topological recursion for the extended Ooguri--Vafa partition function) is the main result of this paper.

There are also additional broader contexts where the results of our paper are potentially interesting. One is the general approach to the topological recursion for the hypergeometric KP tau-functions (generating the so-called weighted Hurwitz numbers), due to Alexandrov--Chapuy--Eynard--Harnad. The hypergeometric tau-functions are forced to depend on an additional parameter $\hbar$ that controls the genus expansion and according to~\cite{APSZ} adapts the hypergeometric tau-functions to the theory of $\hbar$-quantization of the dispersionless KP hierarchy~\cite{TT,NZ}.
We prove that the Brini--Eynard--Mari\~no spectral curve and the extended Ooguri--Vafa partition function fit into the Alexandrov--Chapuy--Eynard--Harnad framework once we allow to consider the $\hbar^2$-deformations of the initial data of hypergeometric tau-functions. This makes the case of the Brini--Eynard--Mari\~no spectral curve close to the case of the so-called $r$-spin Hurwitz numbers, where the effect of $\hbar^2$-deformation of the initial data of hypergeometric tau-functions is also present~\cite{BKLPS,DKPS}, and it raises immediately a lot of open questions on the general role of the $\hbar^2$-deformations in the Alexandrov--Chapuy--Eynard--Harnad framework for the topological recursion for weighted Hurwitz numbers.
These questions are beyond the scope of this paper, but they are considered and partly answered in~\cite{BDKS,BDKS-2}.


Another context is as follows: 
there is a mirror symmetry statement of remodelling type proved recently by Fang--Zong in~\cite{FangZong}. They show that the open Gromov--Witten invariants of the resolved conifold with the special choice of Langrangian corresponding to a torus knot~\cite{DiaconescuShendeVafa} are contained in the aforementioned differentials $\omega_{g,n}$, and are obtained by exactly the same symmetrization procedure as the colored HOMFLY-PT polynomials (and thus coincide with the colored HOMFLY-PT polynomials of the torus knot used in the construction of the Lagrangian).
To this end it is very interesting whether the extended Ooguri--Vafa partition function has any geometric counterpart. At the moment neither of the two geometric approaches considered in~\cite{FangZong} (localization formula for the space of maps of bordered Riemann surfaces and relative Gromov--Witten theory) has any obvious refinement that would match the extended Ooguri--Vafa partition function. However, what we get for free is the explicit form of the ELSV-type formula for the $n$-point functions of the Ooguri--Vafa partition function, cf.~\cite[p.~35]{FangZong} (we recall that there is always an ELSV-type formula for the expansions of the correlations differentials obtained via topological recursion in terms of the tautological classes and cohomological field theories on the moduli space of curves~\cite{Eynard,DOSS}).

\subsection{Organization of the paper} Section~\ref{sec:definitionsandtheorems} contains all preliminaries on the formalism of free fermions that we use throughout the paper, the definition of the (extended) Ooguri--Vafa partition functions, a detailed explanation of the statement on topological recursion, its proof modulo the technical results from the rest of the paper, and a discussion of the embedding of the results that we obtain into the Alexandrov--Chapuy--Eynard--Harnad framework for topological recursion for general weighted Hurwitz numbers. 

Section~\ref{sec:EnergyOperator} discusses an evolutionary equation for the extended Ooguri--Vafa partition function obtained from the energy operator, and the analytic extension of the resulting equation for the $n$-point functions of the extended Ooguri--Vafa partition function to the whole spectral curve of Brini--Eynard--Mari\~no. 

In Section~\ref{sec:QLE} we use the equations obtained in the previous section to prove the quadratic loop equations, which is the main technical result of this paper used in Section~\ref{sec:definitionsandtheorems} in the proof of the topological recursion.

\subsection{Acknowledgments} The authors would like to thank B.~Bychkov and A.~Mironov for useful discussions and B.~Fang for useful correspondence. 

S.~S. was supported by the Netherland Organization for Scientific Research.
P.D.-B. and A. S. were supported by the Russian Science Foundation
(Grant No. 20-71-10073).

This project has started when S.~S. was visiting the Faculty of Mathematics at the National Research University Higher School of Economics and the Laboratory of Mathematical and Theoretical Physics at the Moscow Institute of Science and Technology, and S.~S. would like to thank both institutions for warm hospitality and stimulating research atmosphere.

\section{Extended Ooguri--Vafa partition function and topological recursion} \label{sec:definitionsandtheorems}

In this Section, we first briefly recall some basic definitions and facts from the semi-infinite wedge formalism, since we use it 
below in various definitions and computations. Then we recall the definition of the extended Ooguri--Vafa partition function from~\cite{DPSS} that contains, in particular, the full information about the colored HOMFLY-PT polynomials of the torus $(P,Q)$ knot. We also recall the conjecture on the spectral curve 
topological recursion for the extended Ooguri--Vafa partition function posed in~\cite{DPSS}. We explain how this conjecture fits into the general proposal of Alexandrov--Chapuy--Eynard--Harnad~\cite{ACEH-1,ACEH-2} on topological recursion for the partition functions of weighted double Hurwitz numbers (also known as KP tau functions of hypergeometric type). This allows us to formulate (in Subsection \ref{sec:reform}) the main theorem of this paper (the preceding subsections can be seen as the motivation for it), which is a reformulation in ACEH terms of the above-mentioned conjecture. 
Finally, in Subsection \ref{sec:ProofOfTR}, we present a proof of the main theorem 
modulo the technical results on the loop equations obtained in the rest of the paper. 
 
\subsection{Semi-infinite wedge formalism} \label{sec:sinfw}
Let us briefly recall some basic facts from the semi-infinite wedge formalism which is going to be used in what follows (for more details see~\cite{Kac,MiwaJimboDate}).

Let $V$ be an infinite dimensional vector space with a basis labeled by the half integers; denote the basis vector labeled by $m/2$ by $\underline{m/2}$, so $V = \bigoplus_{i \in \Z + \frac{1}{2}} \underline{i}$.
The semi-infinite wedge space $\cV$ is spanned by all wedge products of the form
\begin{equation}\label{wedgeProduct}
	\underline{i_1} \wedge \underline{i_2} \wedge \cdots
\end{equation}
for any decreasing sequence of half integers $(i_k)$ such that there is an integer $c$ (called the charge) with $i_k + k - \frac{1}{2} = c$ for $k$ sufficiently large. We denote the inner product for which this basis is orthonormal by~$(\cdot,\cdot)$. The zero charge subspace $\cV_0\subset\cV$ of the semi-infinite wedge space is spanned by the vectors $v_\lambda$ indexed by integer partitions $\lambda=(\lambda_1\geq\cdots\geq \lambda_\ell)$, $\ell\geq 0$, defined as 
\begin{equation}
	v_\lambda \coloneqq \underline{\lambda_1 - \frac{1}{2}} \wedge \underline{\lambda_2 - \frac{3}{2}} \wedge \underline{\lambda_3 - \frac{5}{2}} \wedge\cdots\,.
\end{equation}	
In particular, the vector $v_{\emptyset} = \underline{-\frac{1}{2}} \wedge \underline{-\frac{3}{2}} \wedge \cdots \in \cV_0$ that corresponds to the empty partition is called the vacuum vector and usually denoted by $|0\rangle$

For an operator $\mathcal{P}$ on $\cV_0$ we define the \emph{vacuum expectation value} of~$\mathcal{P}$ by
$\left\langle \mathcal{P} \right\rangle := \langle 0 |\mathcal{P}|0\rangle$,
where $\langle 0 |$ is the dual of the vacuum vector with respect to the inner product~$(\cdot,\cdot)$ and called the covacuum vector. We will also refer to these vacuum expectation values as (disconnected) \emph{correlators}.
 
For $k\in\mathbb{Z}+\frac{1}{2}$ define the operator $\psi_k\colon \cV\to\cV$ by
$\psi_k \colon (\underline{i_1} \wedge \underline{i_2} \wedge \cdots) \ \mapsto \ (\underline{k} \wedge \underline{i_1} \wedge \underline{i_2} \wedge \cdots)$. It increases the charge by $1$. The operator $\psi_k^*$ is defined to be the adjoint of the operator $\psi_k$ with respect to the inner product~$(\cdot,\cdot)$. Then for $i,j\in\mathbb{Z}+\frac{1}{2}$ define the operator $E_{ij}$ as $E_{ij}\coloneqq{:}\psi_i \psi_j^*{:}$, where the normal order product ${:}\psi_i \psi_j^*{:}$ is defined as follows:
\begin{equation}\label{eq:Eoperdef}
	{:}\psi_i \psi_j^*{:}\ \coloneqq \begin{cases}\psi_i \psi_j^* & \text{ if } j > 0\ ; \\
		-\psi_j^* \psi_i & \text{ if } j < 0\ .\end{cases} 
\end{equation}
The operator $E_{ij}$ does not change the charge and can be restricted to $\cV_0$. 

We will also need the following operators:
\begin{align}		
	&{\mathcal F}_1 := \sum_{k\in\Z+\frac12} k E_{k,k};
	&&{\mathcal F}_2 := \sum_{k\in\Z+\frac12} \frac{k^2}{2} E_{k,k};
	&&\alpha_m:=\sum_{k \in \Z + \frac{1}{2}} E_{k-m,k},\quad m\in \mathbb{Z}.
\end{align}


\subsection{Extended Ooguri--Vafa partition function}

Let us briefly recall the definition of the extended Ooguri--Vafa partition function, following \cite{DPSS} (we refer the reader to that paper for more details and motivations).

\subsubsection{Rosso--Jones formula and Ooguri--Vafa parition function}
	Let $T[Q,P]$ be a torus knot. Then, for a given Young diagram $R$, its $R$-colored HOMFLY-PT polynomial in the spectral framing is defined by the following variation of the Rosso--Jones formula \cite{RossoJones}:
	\begin{equation}
		H_R(T[Q,P];  \mrq,A)	:= A^{P|R|} \cdot 
		\sum_{R_1 \vdash Q|R|} c^{R_1}_R\, \mrq^{2\kappa_R \frac{P}{Q}}\, s^*_{R_1}(\mrq,A) \ ,
	\end{equation}
	where $|R|$ is the size of Young diagram $R$ (the number of boxes in it); the functions $s_R=s_R(p_1,\dots,p_{|R|})$ are the Schur polynomials written in $p$-variables, and $s^*_R:=s_R|_{\forall i\, p_i=p^*_i}$, with
\begin{align}
\label{toplocus}
p^*_i := \frac{A^i - A^{-i}}{\mrq^{i} - \mrq^{-i}} \,;
\end{align} 
	the coefficients $c^{R_1}_R$ are integer numbers determined by the following equation:
	\begin{align}\label{adams}
		s_R(p_{Q},p_{2Q},\dots,p_{(|R|-1)Q},p_{|R|Q})
		= \sum_{R_1 \vdash Q|R|}
		c^{R_1}_{R} \cdot s_{R_1}(p_{1},p_2,\dots,p_{|R_1|});
	\end{align}
and $\kappa_R\coloneqq \sum_{(i,j)\in R} (i-j)$ is the eigenvalue of $\cF_2$ on $v_R$ (the summation runs over the boxes in the Young diagram of the partition $R$).

Now for a moment let us omit the restriction 
\eqref{toplocus}  and define the extended colored HOMFLY-PT polynomial: 
	\begin{equation}
		H_R(T[Q,P];p)	:= A^{P|R|} \cdot \sum_{R_1 \vdash Q|R|} c^{R_1}_R\, \mrq^{2\kappa_R \frac{P}{Q}}\, s_{R_1}(p) \ .
	\end{equation}
Here when writing $p$ in the arguments we assume the whole sets of variables $p_1,\dots,p_{|R|}$ and $p_1,\dots,p_{|R_1|}$ on the LHS and on the RHS, respectively.

	The Ooguri--Vafa partition function~\cite{OV} for the torus knot $T[Q,P]$ depends on two sets of variables, $p=\{p_1,p_2,\dots \}$ and $\pb = \{\pb_1,\pb_2,\dots\}$, and is defined as
	\begin{align} 
		\label{eq:ov-partition-function}
		Z(p, \pb) := \sum_R  H_R(T[Q,P];p) \, s_R(\pb) \ .
	\end{align}
The sum goes over all Young diagrams of all sizes, including the empty one (in that case the corresponding term in the sum is just equal to $1$); and by writing $\pb$ in the argument of $s_R$ on the RHS we mean the restriction of the whole infinite set of variables $\{\pb_1,\pb_2,\dots\}$ to the number of arguments in $s_R$, i.e. $\{\pb_1,\dots,\pb_{|R|}\}$.

Theorem~3.7 of \cite{DPSS} (derived from the representation of the Ooguri--Vafa partition function obtained in~\cite{MMS13}) states that the Ooguri--Vafa partition function can be expressed as follows in the semi-infinite wedge formalism: 
\begin{align} \label{eq:ov-rewritten}
	& Z(p, \pb)  = 
	 \left\langle
	\exp \lb \sum_{j=1}^\infty \frac{\alpha_{-j} p_{j}}{j} \rb
	\exp \lb \hbar \frac{P}{Q} \Fc_2 \rb
	\exp \lb \sum_{i=1}^\infty \frac{\alpha_{iQ} (\pb_i \cdot QA^{iP})}{iQ} \rb
	\right\rangle \ ,
\end{align}
where 
\begin{equation}
\hbar = 2 \log q.
\end{equation}

\subsubsection{Extended Ooguri--Vafa partition function}
Now we define the \emph{extended} Ooguri-Vafa partition function, following~\cite{DPSS}, which extends $Z$ to $\pb$'s with fractional indices:
\begin{definition}\label{def:extov}
The extended Ooguri-Vafa partition function is
\begin{align} \label{eq:extended-partition-function}
	& Z^{\ext}(p, \pb) \coloneqq
	\left\langle
	\exp \lb \sum_{j=1}^\infty \frac{\alpha_{-j} p_{j}}{j} \rb
	\exp \lb \hbar \frac{P}{Q} \Fc_2 \rb
	\exp \lb \sum_{i=1}^\infty \frac{\alpha_{i} (\pb_{i/Q} QA^{iP/Q})}{i} \rb
	\right\rangle \ .
\end{align}
Here $p=\{p_1,p_2,\dots\}$, $\pb=\{p_{i/Q}|i\in \mathbb{Z}_{>0}\}$.
\end{definition}

\begin{remark}
Note that $Z$ is recovered from $Z^{\ext}$ by keeping only $\pb_{i/Q}$ with integer indices, i.e.
\begin{equation}
	Z(p_1,p_2,\dots ; \pb_1,\pb_2,\dots) = Z^{\ext}(p_1,p_2,\dots ; \pb_{1/Q},\pb_{2/Q},\dots) |_{\forall i, Q\nmid i:\, p_{i/Q}=0}.
\end{equation}
\end{remark}

\begin{remark}The reason to define such an extension is that this extension takes the form of the generating function of double Hurwitz numbers, after the following change of variables:
\begin{align}\label{eq:ptpbrel}
	\pt_i &:= \pb_{i/Q} QA^{iP/Q}, \quad i\geq 1 \ ;\\
	u &:= \dfrac{P}{Q}\hbar \ .
\end{align}
In terms of these variables we get:
\begin{align}\label{eq:Zextpt}
Z^{\ext}(p,\pt) = 
\left\langle
\exp \lb \sum_{i=1}^\infty \frac{\alpha_{i} p_i}{i} \rb
\exp \lb u \Fc_2 \rb
\exp \lb \sum_{j=1}^\infty \frac{\alpha_{-j} \pt_{j}}{j} \rb
\right\rangle \ ,
\end{align}
which is precisely the standard form of the generating function of double Hurwitz numbers.
\end{remark}

We are actually interested in the restriction of $Z^{\ext}(p,\pb)$ to the topological locus, which is a physical jargon for the substitution $p=p^*$, where
\begin{align} \label{eq:pstarhbar}
	p^*_i = \frac{A^i - A^{-i}}{e^{\hbar i/2} - e^{-\hbar i/2}} \,.
\end{align}
Denote
\begin{equation}\label{eq:toploc}
	Z^{\ext} = Z^{\ext}(\pb):=Z^{\ext}(p, \pb)\big|_{p=p^*} \ .
\end{equation}

\subsubsection{The $n$-point functions}
Now let us define the $n$-point functions of the extended Ooguri-Vafa partition function. 
Let 
\begin{align}  \label{eq:definitionCgmu}
	C^{(g)}_{\mu_1 \dots \mu_n} & := [\hbar^{2g-2+n}] \dfrac{\p}{\p\pb_{\mu_1/Q}}\dots \dfrac{\p}{\p\pb_{\mu_n/Q}}\log Z^{\ext}\bigg|_{\pb=0}
	\ .
\end{align}
\begin{remark}
	Note that after the restriction~\eqref{eq:toploc}, the $\hbar$ parameter enters not only the middle exponential in~\eqref{eq:extended-partition-function}, but also the one on the left.
\end{remark}

Now, following~\cite{DPSS}, we have
\begin{definition}
	The $n$-point functions $H_{g,n}$ of the extended Ooguri--Vafa partition function are
	\begin{align} \label{eq:HgnFirstTime}
		H_{g,n}(\Lambda_1,\dots,\Lambda_n) & := Q^{-n} \sum_{\mu_1 \dots \mu_n = 1}^\infty 
		C^{(g)}_{\mu_1 \dots \mu_n} \prod_{i = 1}^n \Lambda_i^{-\mu_i}\, .
	\end{align}
\end{definition}

\subsection{Spectral curve 
	topological recursion}

\subsubsection{Formulation of topological recursion}
The spectral curve topological recursion 
\cite{CE1,EynardOrantin,OrantinThesis,EynardSurvey} 
associates to a Riemann surface $\Sigma$  (the so-called spectral curve)  equipped with two functions $X,y\colon \Sigma\to \mathbb{C}$ and a symmetric bidifferential $B$ on $\Sigma^2 = \Sigma\times\Sigma$ satisfying some extra conditions a family of meromorphic symmetric $n$-differentials (CEO-differentials) $\omega_{g,n}$ defined on $\Sigma^n$, $g\geq 0$, $n\geq 1$. We assume that $dX$ is meromorphic and all critical points $p_1,\dots,p_r$ of $X$ are simple, $y$ is holomorphic near $p_i$ and $dy\not=0$ at $p_i$, $i=1,\dots,r$, and $B$ has no singularities except for a double pole on the diagonal with biresidue $1$. We set by definition $\omega_{0,1} =y\, dX/X$, $\omega_{0,2} = B$, and for $2g-2+n>0$ we define:
\begin{align}\label{eq:CEO-toprec}
	\omega_{g,n}(z_{\{1,\dots,n\}}) = 
	\frac 12 \sum_{i=1}^r \Res_{z\to p_i} \frac{\int_z^{\sigma_i(z)}\omega_{0,2}(\cdot,z_1)}{\omega_{0,1}(\sigma_i(z))-\omega_{0,1}(z)} \Big[
	\omega_{g-1,n+1}(z,\sigma_i(z),z_{\{2,\dots,n\}}) + 
	\\ \notag 
	\sum_{\substack{g_1+g_2 = g \\ I_1\sqcup I_2 = \{2,\dots,n\} \\ (g_i,|I_i|)\not= (0,0)}}
	\omega_{g_1,1+|I_1|}(z,z_{I_1})\omega_{g_2,1+|I_2|}(\sigma_i(z),z_{I_2})
	\Big].
\end{align}
Here $z_i$ is some coordinate on the $i$-th copy of $\Sigma$ in $\Sigma^n$, $\sigma_i$ is the deck transformation for $X$ near the point $p_i$, $i=1,\dots,r$, and all $\omega_{-1,n}$, $n\geq 1$, are set to be equal to $0$. Furthermore, for a set $I$, we write $ z_I = \{ z_i \}_{i \in I}$.

\begin{remark}\label{rem:adjustment01} Note that we can add to $\omega_{0,1}$ any differential 
	invariant under the deck transformations $\sigma_i$, $i=1,\dots,r$, and this does not change any $\omega_{g,n}$ for $2g-2+n\geq 0$. 
\end{remark}

\subsubsection{Main theorem}



Let us state the main theorem of the present paper (which was formulated in \cite{DPSS} as a conjecture). 
\begin{theorem} \label{thm:Main} For $2g-2+n\geq 0$, the
	correlation differentials of the extended Ooguri-Vafa partition function 
	\begin{align}
		\omega_{g,n} & := d_1\cdots d_n H_{g,n} = Q^{-n} \sum_{\mu_1 \dots \mu_n = 1}^\infty 
		C^{(g)}_{\mu_1 \dots \mu_n} \prod_{i = 1}^n (-\mu_i) \frac{d \Lambda_i}{\Lambda_i^{\mu_i + 1}} \ ,
	\end{align}
satisfy the CEO spectral curve topological recursion on the spectral curve $\Sigma = \mathbb{C}P^1\setminus C$ with the global coordinate $U$, where $C$ is a contour connecting the points $U=A^{-P/Q-1}$ and $U=A^{-P/Q+1}$, for the data $X=X^\mathrm{BEM}$, $y=y^\mathrm{BEM}$, and $B=dU_1dU_2/(U_1-U_2)^2$, where
\begin{align}\label{eq:LBEM}
	\Lambda(U) & := U \left (\dfrac{ 1 - A^{P/Q+1} U}{ 1 - A^{P/Q-1} U}\right )^{P/Q} \ ,\\
	X^\mathrm{BEM} & \coloneqq \Lambda^{-Q}(U) \ ,\\
	 \label{eq:yBEM}
	y^\mathrm{BEM} & \coloneqq 
	\gamma \log(U)
	+\delta\log \left ( 1 - A^{P/Q+1} U\right ) - \delta \log \left ( 1 - A^{P/Q-1} U\right ) \, ,
\end{align}
and $\gamma$ and $\delta$ are chosen such that
\begin{align}
	\left(
	\begin{array}{cc}
		Q & P \\
		\gamma & \delta
	\end{array}
	\right)
	\in \mathrm{SL}(2,\mathbb{Z}) \ .
\end{align}
\end{theorem}
We will see below that this theorem (or, more precisely, its equivalent reformulation introduced in the next subsection) follows from Theorem~\ref{thm:QLE} (which is the main technical result of the present paper) taken together with the results of \cite{DPSS}.

\begin{remark} In the case $(g,n)=(0,1)$,  the correlation differential $\omega_{0,1}=d_1H_{g,n}$ is not equal to $y^\mathrm{BEM}\, dX^\mathrm{BEM}/X^\mathrm{BEM}$, but the difference 
\begin{equation}\label{eq:Omega01FirstTime}
\omega_{0,1}-y^\mathrm{BEM} \frac{dX^\mathrm{BEM}}{X^\mathrm{BEM}} = \left(\frac{\gamma}{Q} \log X^\mathrm{BEM} - \frac 1Q\log A^2 \right) \frac{dX^\mathrm{BEM}}{X^\mathrm{BEM}} 
\end{equation}	
 is invariant under all deck transformations, as is allowed by Remark~\ref{rem:adjustment01}. See~\cite[Section 9.1]{DPSS}.
\end{remark}

\begin{remark}
	The choice $X=\Lambda^{-Q}$ differs from what was stated in \cite{DPSS} (where $X=\Lambda^Q$ was used). However, in the formula for topological recursion~\eqref{eq:CEO-toprec} this just amounts to taking the opposite sign in for $\omega_{0,1}$, which corresponds to the fact that in the present paper we just use the definition of topological recursion with a different sign convention as opposed to~\cite[Equation (3.55)]{BEM11} and~\cite{DPSS}. The reason for the change of the sign convention is the link to the Alexandrov--Chapuy--Eynard--Harnad theory that we establish below. 
\end{remark}

\begin{remark}\label{rem:Agen}
	Note that for the purpose of definition of the spectral curve topological recursion, we have to treat $A$ as some fixed complex number, rather than a formal variable as in the definition of the HOMFLY-PT polynomials. For the purposes of the present paper, $A$ can be an arbitrary complex number, apart from the mild restrictions requiring that the zeroes of $d\Lambda$ should be finite and simple and should not have the same absolute values as either $A$ or $A^{-1}$. Note that $d\Lambda(U) =0$ is equivalent to 
	\begin{equation}
	(A^{-P/Q} U^{-1})^2 - \Big(A\big(1+\frac PQ\big)+ A^{-1}\big(1-\frac PQ\big)\Big) (A^{-P/Q} U^{-1}) +1 =0,		
	\end{equation}
	so the first condition just means that at the zeros of $d\Lambda$ the coordinate $z\coloneqq A^{-P/Q} U^{-1}$ is not equal to $0$, $\infty$, $1$, or $-1$. We can also assume without loss of generality that $|A| < |A^{-1}|$.
\end{remark}

\subsection{Alexandrov--Chapuy--Eynard--Harnad framework}\label{sec:ACEHintro}  In \cite{ACEH-1,ACEH-2} a rather general Hurwitz-type problem was considered (specifically, polynomially weighted double Hurwitz numbers), and it was proved that the corresponding $n$-point functions satisfy topological recursion on a specific spectral curve determined by the way the numbers are weighted. It turns out that while the results of these papers are not directly applicable to our case, it is still possible to use their setup, in particular, their formula for the spectral curve. 

\subsubsection{Functions $X$ and $y$}
Specifically, for a Hurwitz-type problem coming from the partition functions \eqref{eq:toploc}, the function $X$ and $y$ on the ACEH spectral curve (cf. also \cite{BDKS}) take the form 
\begin{align}\label{eq:XACEH}
X^\mathrm{ACEH} &= z\,\phi(\check S(z))^{-1}\\ \label{eq:yACEH}
y^\mathrm{ACEH} &= \check S(z),
\end{align}
where
\begin{align}	
	\phi(y) &= \exp\Big(\frac PQ y\Big),\\
	\check S(z) &= \sum_{j=1}^\infty p^*_j \vert _{\hbar = 0 } z^j = \sum_{j=1}^\infty \frac{A^j-A^{-j}}{j } z^j = -\log (1-Az) + \log (1-A^{-1}z).
\end{align}

Let us determine how the ACEH-functions $X^\mathrm{ACEH}$ and $y^\mathrm{ACEH}$ are related to $X^\mathrm{BEM}$ and $y^\mathrm{BEM}$.

\begin{proposition} For $U = A^{-P/Q} z^{-1}$, we have:
	\begin{align}
		X^\mathrm{ACEH}(z) & = A^{P/Q}\left(X^\mathrm{BEM}(U)\right)^{1/Q}, \\
		\label{eq:yACEHBEM}
		y^\mathrm{ACEH}(z) & = 		Qy^\mathrm{BEM}(U) +\gamma \log X^\mathrm{BEM}(U) - \log A^2
	\end{align}
\end{proposition} 

\begin{proof} Indeed, 
	\begin{align}
	& A^{P/Q}\Lambda^{-1}  =  A^{P/Q} U^{-1} \left (\dfrac{ 1 - A^{P/Q-1} U}{ 1 - A^{P/Q+1} U}\right )^{P/Q} = A^{2P/Q} z \left (\dfrac{ 1 - A^{-1} z^{-1}}{ 1 - A z^{-1}}\right )^{P/Q} \\ \notag
	& = z \left (\dfrac{ 1- A z  }{ 1- A^{-1} z  }\right )^{P/Q}  = z \exp \lb - \frac PQ  \lb  -\log (1-Az) + \log (1-A^{-1}z)\rb\rb
	\end{align}
and
	\begin{align}
		&Q (y^{BEM}(U) - \gamma \log \Lambda(U)) - \log A^2 \\ \notag
		& =Q\left(\gamma \log U 
		+\delta\log \left ( 1 - A^{P/Q+1} U\right ) - \delta \log \left ( 1 - A^{P/Q-1} U\right )\right) \\ \nonumber
		&\hspace{1cm}- \gamma \log U -\dfrac{\gamma P}{Q} \log(1 - A^{P/Q+1} U) +\dfrac{\gamma P}{Q} \log( 1 - A^{P/Q-1} U)) - \log A^2 \\ \nonumber
		&=\log \left ( 1 - A^{P/Q+1} U\right ) - \log \left ( 1 - A^{P/Q-1} U\right )- \log A^2 \\ \nonumber
		&= \log \left ( 1 - A z^{-1}\right) - \log \left ( 1 - A^{-1} z^{-1}\right )- \log A^2 \\ \nonumber
		&=	\log\left(\dfrac{1-A^{-1}z}{1-Az}\right).   
	\end{align}
\end{proof}

\begin{corollary} We have
	\begin{equation}
	\omega_{0,1} = y^\mathrm{ACEH} \dfrac{dX^\mathrm{ACEH}}{X^\mathrm{ACEH}}.
	\end{equation}
\end{corollary}
\begin{proof}
Indeed,
\begin{align}
y^\mathrm{ACEH} \dfrac{dX^\mathrm{ACEH}}{X^\mathrm{ACEH}} = 
\left( Qy^\mathrm{BEM}(U) +\gamma \log X^\mathrm{BEM}(U) - \log A^2 \right) \frac 1Q \dfrac{dX^\mathrm{BEM}}{X^\mathrm{BEM}},
\end{align}
and the latter expression is equal to $\omega_{0,1}$ by Equation~\eqref{eq:Omega01FirstTime}.
\end{proof}

Now let us determine which form does the definition \eqref{eq:HgnFirstTime} of the $n$-point functions $H_{g,n}$ take in the ACEH coordinates.
Using the fact that $X^\mathrm{ACEH} = A^{P/Q} \Lambda^{-1}$ and Equation~\eqref{eq:ptpbrel}, we can rewrite Equation~\eqref{eq:HgnFirstTime} as 
\begin{align}\label{eq:HXdef}
	H_{g,n}
	 & = Q^{-n}A^{-nP/Q} \sum_{\mu_1 \dots \mu_n = 1}^\infty 
	C^{(g)}_{\mu_1 \dots \mu_n} \prod_{i = 1}^n \left(X^\mathrm{ACEH}_i\right)^{\mu_i} \\ \notag
	& = \sum_{\mu_1 \dots \mu_n = 1}^\infty [\hbar^{2g-2+n}] \dfrac{\p}{\p\tilde p_{\mu_1}}\dots \dfrac{\p}{\p\tilde p_{\mu_n}}\log Z^{ext}\bigg|_{\tilde p=0} \prod_{i = 1}^n \left(X^\mathrm{ACEH}_i\right)^{\mu_i}
	.
\end{align}
The latter form is the standard one in the ACEH theory. Note that in the last line the derivatives are with respect to $\pt$-variables defined in \eqref{eq:ptpbrel} as opposed to the $\pb$-variables used in \eqref{eq:definitionCgmu}.

\subsection{Reformulation of topological recursion for the correlation differentials of the extended Ooguri--Vafa partition function} \label{sec:reform}

In this subsection we can finally reformulate Theorem~\ref{thm:Main} in the ACEH form, in which 
we are going to prove it further in the paper.

\begin{convention}
From now on and for the rest of the paper, we adopt the convention that $X$ and $y$ 
stand for
$X^{\mathrm{ACEH}}$ and $y^{\mathrm{ACEH}}$, respectively.	
\end{convention}

Recall formulas \eqref{eq:Zextpt}-\eqref{eq:pstarhbar} and \eqref{eq:HXdef} for the extended Ooguri-Vafa partition function and the corresponding $n$-point functions respectively, in $\pt$-variables:
\begin{align} \label{eq:Zextff}
	Z^{\ext}(\pt) &= 
	\left\langle
	\exp \lb \sum_{i=1}^\infty \frac{\alpha_{i} p^*_i}{i} \rb
	\exp \lb \frac{P}{Q}\hbar \Fc_2 \rb
	\exp \lb \sum_{j=1}^\infty \frac{\alpha_{-j} \pt_{j}}{j} \rb
	\right\rangle \ ,\\	
	p^*_i &= \frac{A^i - A^{-i}}{e^{\hbar i/2} - e^{-\hbar i/2}} \,, \\ \label{eq:Hgnff}
	H_{g,n}(X_1,\dots,X_n) & = \sum_{\mu_1 \dots \mu_n = 1}^\infty [\hbar^{2g-2+n}] \dfrac{\p}{\p\tilde p_{\mu_1}}\dots \dfrac{\p}{\p\tilde p_{\mu_n}}\log Z^{ext}\bigg|_{\tilde p=0} \prod_{i = 1}^n X_i^{\mu_i}\, .
\end{align}

Now, the results of Section \ref{sec:ACEHintro} imply that 
Theorem~\ref{thm:Main} can be equivalently reformulated as follows: 
\begin{theorem} The correlation differentials of the extended Ooguri--Vafa partition function 
	\begin{align}\label{eq:omegadef}
		\omega_{g,n} & := d_1\cdots d_n H_{g,n}
	\end{align}
	satisfy 
	topological recursion for the spectral curve data $\Sigma = \mathbb{C}P^1\setminus C$ with the global coordinate $z$, where $C$ is a contour connecting the points $z=A$ and $z=A^{-1}$, equipped with the functions 
	\begin{align}	\label{eq:SpectralCurve-Definition-X-y}
		X(z) & = z\exp\left(-\frac{P}{Q}\,y(z)\right), \\
		y(z) &= 
		\log\left(\dfrac{1-A^{-1}z}{1-Az}\right),
	\end{align}
and the symmetric bidifferential $B=dz_1dz_2/(z_1-z_2)^2$. 
\end{theorem}

The remainder of the paper is devoted to proving this theorem. 

\begin{remark} Note that compared to the theory developed in~\cite{ACEH-1,ACEH-2}, in our case 
\begin{equation}	
\sum_{i=1}^\infty p_i^*z^i = \frac{\hbar z\frac{\p}{\p z}}{\zeta(\hbar z\frac{\p}{\p z})} y(z) = y(z)+ O(\hbar^2),
\end{equation}
 where $\zeta(t)\coloneqq e^{t/2}-e^{-t/2}$, as opposed to the standard ACEH convention that $\sum_{i=1}^\infty p_i^*z^i = y(z)$. This means that the extended Ooguri--Vafa partition functions falls in the extension of the ACEH theory, where the 
 data specifying the partition function can depend on $\hbar^2$. Another example where the ACEH theory computes via the topological recursion some $\hbar^2$-deformed partition function is the case of Hurwitz numbers with completed cycles, see~\cite{BKLPS,DKPS}. All these examples inspired a systematic study of the $\hbar^2$-extension of the ACEH theory, see~\cite{BDKS,BDKS-2}.
\end{remark}
 
\subsection{Proof of topological recursion} \label{sec:ProofOfTR} We follow the scheme of the proof of topological recursion suggested in~\cite{BS}, see also~\cite{BEO}. Namely according to~\cite[Theorem 2.2]{BS}, the system of correlation differentials $\omega_{g,n}$ defined on the spectral curve satisfies the topological recursion if and only if it satisfies the projection property, the linear loop equations, and the quadratic loop equations. 

In the case of the correlation differentials of the extended Ooguri--Vafa partition function, the first two properties are superseded by the so-called quasi-polynomiality property, see~\cite[Theorems 7.1 and 8.1]{DPSS} (cf.~Remark 8.2 in \emph{op.~cit.}). So, this way we reduce the proof of topological recursion to the statement that the correlation differentials of the extended Ooguri--Vafa partition function satisfy the quadratic loop equation. Let us formulate this statement explicitly.

\begin{notation} Denote
\begin{equation}\label{eq:Wsdef}
	W_{g,n}(z_{\{1,\dots,n\}}) := \omega_{g,n}(z_{\{1,\dots,n\}}) / \prod_{j=1}^n d\log X(z_j). 
\end{equation}
\end{notation}

\begin{definition}
We say that a system of meromorphic differentials $\omega_{g,n}$ 
satisfies the \emph{quadratic loop equations} if for any $g\geq 0$, $n\geq 0$ the expression
\begin{equation}\label{eq:QLE1}
	W_{g-1,n+2}(z,\sigma_i(z),z_{\{1,\dots,n\}}) + \sum_{\substack{g_1+g_2 = g \\ I_1\sqcup I_2 = \{1,\dots,n\} }}
	{W}_{g_1,1+|I_1|}(z,z_{I_1}){W}_{g_2,1+|I_2|}(\sigma_i(z),z_{I_2})
\end{equation}
is holomorphic in $z$ for $z\to p_i$, where $p_i$ are the zeroes of $dX$.
\end{definition}


\begin{notation}
	Let 
	\begin{equation}\label{eq:Ddef}
		D_X:=X\dfrac{d}{dX}.
	\end{equation}
\end{notation}

\begin{theorem} \label{thm:QLE} 
Functions $W_{g,n}$ 
coming from multidifferentials $\omega_{g,n}$ 
of \eqref{eq:omegadef}, i.e. given by
	\begin{equation}
		W_{g,n}(z_{\{1,\dots,n\}}) = D_{X_1}\cdots D_{X_n} H_{g,n},
	\end{equation}
for functions $H_{g,n}$ 
 defined by~\eqref{eq:HXdef}, satisfy the quadratic loop equation \eqref{eq:QLE1}.
\end{theorem}

This theorem is proved in Section~\ref{sec:QLE}. By~\cite[Theorem 2.2]{BS} and~\cite[Theorem 8.1 and Remark 8.2]{DPSS} it implies Theorem~\ref{thm:Main}.

\begin{remark} The same method of proving the topological recursion, that is, a combinatorial-algebraic proof of quasi-polinomiality and a subsequent derivation of the quadratic loop equation is used in a number of papers, see~\cite{BKLPS,DKPS,DKPS-0,KPS}. The technical details of the proofs are, however, quite different, and in our case we have to combine, quite non-trivially, a set of ideas from these references.
\end{remark}

\section{Energy operator and cut-and-join equation} \label{sec:EnergyOperator}

In this section we use the energy operator to derive a cut-and-join equation for the extended Ooguri--Vafa partition function. This type of cut-and-join equations was used earlier in~\cite{GGN} and~\cite{DKPS-0} for the monotone single Hurwitz numbers, and here we  follow very closely a very general computation of the second named author in~\cite{Kaz2020}, see also~\cite{Kaz2019}.

The second goal of this section is to analytically extend the cut-and-join equation (originally derived as an equation for formal power series $X_1,\dots,X_n$ near $X_1=\cdots =X_n =0$) to an equation for the globally define $n$-point functions.  

\subsection{Vertex operators and an equation for the partition function} \label{sec:genpart1}
Let us for the duration of Sections \ref{sec:genpart1} and \ref{sec:genpart2} forget about the specifics of the (extended) Ooguri-Vafa partition function and consider a general partition function of the form
\begin{align} \label{eq:GeneralZ}
Z= 
\bigcovac
\exp \lb \sum_{i=1}^\infty \frac{\alpha_{i} p_i}{i} \rb
\Dop(\hbar)
\exp \lb \sum_{j=1}^\infty \frac{\alpha_{-j} s_j}{j\hbar} \rb
\bigvac,
\end{align}
where $\Dop(\hbar) v_\lambda = \prod_{(i,j) \in \lambda} \phi(\hbar (i-j))$, 
where the indices $(i,j)$ of boxes of Young diagram $\lambda$ run as follows: $j=1,\dots,\ell(\lambda)$, $i=1,\dots,\lambda_j$; and $\phi(y):=1+\sum_{k=1}^\infty \varphi_k y^k$ is some formal power series.  
Define the coefficients $c_{r,j}$, $r\geq 0$, $j\geq 1$, by
\begin{align} \label{eq:Definition-Crj}
\prod_{k=\ell+\frac 12}^{\ell+j-\frac 12} \phi(\hbar k) = \sum_{r=0}^\infty \frac{c_{r,j}\hbar^r}{r!} \lb \ell+\frac j2 \rb^{r}.
\end{align}
Define the differential operators
\begin{equation}
\hat \cQ_{r,m}  \coloneqq
[u^r] \,  \sum_{\substack{n+k\geq 1 \\ a_1,\dots,a_n\geq 1 \\ b_1,\dots,b_k\geq 1 \\ \sum a_i = \sum b_i+m}} \frac{1}{\zeta(u)}
\frac{1}{n!k!} \Bigg[\prod_{i=1}^n \frac{\zeta(ua_i)}{a_i} p_{a_i} \Bigg] \Bigg[\prod_{i=1}^k \frac{\zeta(ub_i)}{b_i} \frac{b_i\partial}{\partial p_{b_i}} \Bigg]\,,
\end{equation}
$r\geq 0$, $m\in \mathbb{Z}$. The simplest cases are $\hat{\cQ}_{0,-m} =  \frac{m\partial}{\partial p_m}$ and $\hat{\cQ}_{0,m} =  p_{m}$ for $m>0$.

\begin{lemma} \label{lem:Cut-And-Join-P-Op} We have:
	\begin{align}
	\sum_{i=1}^\infty {p_i} \frac{i\partial}{\partial p_i} Z =
	\sum_{m=1}^\infty \frac{s_m}{\hbar} \sum_{r=0}^\infty c_{r,m} \hbar^r \hat \cQ_{r,m} Z.
	\end{align}
\end{lemma}

\begin{proof}
Consider the energy operator $\cF_1 = \sum_{\ell\in\halfZ}\ell E_{\ell\ell}$. Note that 
\begin{align}
& \cF_1 \vac = 0 \,; \\
& \exp \lb \sum_{j=1}^\infty \frac{\alpha_{-j} s_j}{j\hbar} \rb \cF_1 \exp \lb -\sum_{j=1}^\infty \frac{\alpha_{-j} s_j}{j\hbar} \rb
= \cF_1 - \sum_{j=1}^\infty \frac{\alpha_{-j} s_j}{\hbar} \, ; \\
& \Dop(\hbar) \lb \cF_1 - \sum_{j=1}^\infty \frac{\alpha_{-j} s_j}{\hbar} \rb \Dop(\hbar)^{-1} = \cF_1 - \sum_{j=1}^\infty \frac{s_j}{\hbar} \sum_{\ell\in \halfZ} E_{\ell+j,\ell} \prod_{k=\ell+\frac 12}^{\ell+j-\frac 12} \phi(\hbar k) \, .
\end{align}
Using these three equations and Equation~\eqref{eq:Definition-Crj}, we have:
\begin{align}\label{eq:F1-Equation-Operators}
& \bigcovac
\exp \lb \sum_{i=1}^\infty \frac{\alpha_{i} p_i}{i} \rb
\cF_1 
\Dop(\hbar)
\exp \lb \sum_{j=1}^\infty \frac{\alpha_{-j} s_j}{j\hbar} \rb
\bigvac
\\ & \notag
= 
\bigcovac
\exp \lb \sum_{i=1}^\infty \frac{\alpha_{i} p_i}{i} \rb
\lb
\sum_{j=1}^\infty \frac{s_j}{\hbar} \sum_{\ell\in \halfZ} E_{\ell+j,\ell} \sum_{r=0}^\infty \frac{c_{r,j}\hbar^r}{r!} \lb \ell+\frac j2 \rb^{r}
\rb
\Dop(\hbar)
\exp \lb \sum_{j=1}^\infty \frac{\alpha_{-j} s_j}{j\hbar} \rb
\bigvac\,.
\end{align}
Define the operators 
\begin{equation}
\cQ_{r,m} \coloneqq  \sum_{\ell\in \halfZ} E_{\ell+m,\ell} \frac {(\ell+\frac m2)^r}{r!},\ r\geq 0,\ m\in \Z.
\end{equation}
According to~\cite[Chapter 14]{Kac}, for any vector $v_\lambda$ we have 
\begin{align} \label{eq:hat-for-Q}
& \bigcovac \lb \sum_{i=1}^\infty \frac{\alpha_{i} p_i}{i} \rb \cQ_{r,m} v_\lambda =  \hat \cQ_{r,m} \bigcovac \lb \sum_{i=1}^\infty \frac{\alpha_{i} p_i}{i} \rb  v_\lambda.
\end{align}
Applying \eqref{eq:hat-for-Q} to Equation~\eqref{eq:F1-Equation-Operators} we obtain the statement of the lemma.
\end{proof}

\begin{remark} Note that the coefficients $s_m$, $m\geq 1$, are allowed to be formal power series in $\hbar^2$ in the computation above.
\end{remark}

\subsection{Equations for $n$-point functions}\label{sec:genpart2} For a general $Z$ as in~\eqref{eq:GeneralZ} we consider the $n$-point functions 
\begin{align}
H_{g,n} = [\hbar^{2g-2+n}]\sum_{d_1,\dots,d_n} 
\frac{\partial^n \log Z}{\partial {p}_{d_1}\cdots\partial{p}_{d_n}} \bigg |_{{p}=0} \prod_{i=1}^n X_i^{d_i}. 
\end{align}

Let $D_{X_i} = X_i\frac{d}{d X_i}$, as in \eqref{eq:Ddef}. Lemma~\ref{lem:Cut-And-Join-P-Op} implies the following equation for some slight modification of the $n$-point functions.
\begin{corollary} We have: 
\begin{align}
& \lb\sum_{i=1}^n D_{X_i}\rb \tilde H_{g,n} = \\ \notag
& \sum_{m=1}^\infty s_m \sum_{r=0}^\infty c_{r,m} 
\sum_{\substack{t \geq 0, d \geq 0 \\ t + 2d =r }} 
\frac{1}{t!} 
\sum_{\substack{\ell = 1,\dots, t \\ \text{or } \ell=0 \text{ if } t=0} }
\frac{1}{\ell!}
\sum_{\substack{
		\{ k\} \sqcup \bigsqcup_{j=1}^\ell K_j = \llbracket n \rrbracket \\ 
		\bigsqcup_{j=1}^\ell T_j =  \llbracket t \rrbracket\\ 
		T_j \neq \emptyset \\
		g-d = \sum_{j=1}^\ell g_j + t - \ell \\
		g_1,\ldots,g_{\ell} \geq 0 
}}  
Q_{d,m,t}^{(k)} \bigg[\prod_{j = 1}^{\ell} \tilde{H}_{g_j,|T_j|+|K_j|}(\xi_{T_j},X_{K_j})\bigg]
\end{align}
where
\begin{align}
\label{Qrdef} 
\sum_{d \geq 0} Q_{d,m,t}^{(k)}\,z^{2d} & = \frac{z}{\zeta(z)}  \frac{\zeta(zD_{X_k})}{zD_{X_k}} \circ X_k^m \circ \prod_{j = 1}^t \left(\bigg\rfloor_{\xi_j = x_k} \circ \frac{\zeta(zD_{\xi_j})}{z}\right).
\end{align}
The modification of $n$-point functions that we have to use is the following. For $(g,n)\not=(0,2)$ we set $\tilde H_{g,n} = H_{g,n}$. For $(g,n) = (0,2)$ we set
\begin{align}
\tilde H_{0,2} (\xi_i,\xi_j) & \coloneqq H_{0,2} (\xi_j,\xi_j) ; \\ \notag
\tilde H_{0,2} (\xi_j,X_j) & \coloneqq H_{0,2} (\xi_i,X_j) +\log \left( \frac{\xi_i - X_j}{\xi_iX_j}\right)\,
\end{align}
(here we have slightly abused the notation, treating $\tilde H_{0,2}$ with two $\xi$-arguments differently from $\tilde H_{0,2}$ with one $\xi$-argument and one $X$-argument; note that there is at least one $\xi$-argument in every $\tilde H_{g,n}$ entering \eqref{Qrdef}).
\end{corollary}

\begin{proof} The proof literally repeats step by step the proof of~\cite[Proposition 13]{BKLPS}. See also a different exposition in~\cite{Kaz2020}.
\end{proof}

\begin{remark}
Assume now that $s_m=\sum_{a=0}^\infty s_{m,a}\hbar^{2a}$. Then the formula above gets a correction, namely 
\begin{align} \label{eq:cutjoin}
& \lb\sum_{i=1}^n D_{X_i}\rb \tilde H_{g,n} = \\ \notag
& \sum_{\substack{m\geq 1\\a\geq 0}}^\infty s_{m,a} \sum_{r=0}^\infty c_{r,m} 
\sum_{\substack{t \geq 0, d \geq 0 \\ t + 2d =r }} 
\frac{1}{t!} 
\sum_{\substack{\ell = 1,\dots, t \\ \text{or } \ell=0 \text{ if } t=0} }
\frac{1}{\ell!}
\sum_{\substack{
		\{ k\} \sqcup \bigsqcup_{j=1}^\ell K_j = \llbracket n \rrbracket \\ 
		\bigsqcup_{j=1}^\ell T_j =  \llbracket t \rrbracket\\ 
		T_j \neq \emptyset \\
		g-d-a = \sum_{j=1}^\ell g_j + t - \ell \\
		g_1,\ldots,g_{\ell} \geq 0 
}}  
Q_{d,m,t}^{(k)} \bigg[\prod_{j = 1}^{\ell} \tilde{H}_{g_j,|T_j|+|K_j|}(\xi_{T_j},X_{K_j})\bigg]\,
\end{align}
(note the appearance of $-a$ in the LHS of the second line from the bottom in the list of conditions of the large sum).
\end{remark}

\subsection{Equations for the extended Ooguri-Vafa $n$-point functions}  \label{sec:ChangeFromQW-to-CheckH} 
Let's adapt Equation~\eqref{eq:cutjoin} to the case of $Z^{\ext}$ of \eqref{eq:Zextff}. 
Consider the $n$-point functions given in Equation~\eqref{eq:Hgnff}:
\begin{align}
H_{g,n} = [\hbar^{2g-2+n}]\sum_{d_1,\dots,d_n} 
\frac{\partial^n \log Z^{\ext}}{\partial \tilde{p}_{d_1}\cdots\partial\tilde{p}_{d_n}} \bigg |_{\tilde{p}=0} \prod_{i=1}^n X_i^{d_i}\,. 
\end{align}
In our case $\Dop(\hbar) = \exp(\hbar \frac PQ \cF_2)$, which implies that $\phi(\hbar z ) = \exp(\hbar \frac PQ z)$. Therefore, $\prod_{j=\ell+\frac 12}^{\ell+m-\frac 12} \phi(\hbar \frac PQ j) = \exp(m\frac PQ \hbar(\ell+\frac m2))$, which implies that $c_{r,m} = \left(m\frac{P}{Q}\right)^r$ for any $r\geq 0$ and $m\geq 1$. 
The coefficients $s_m$ are equal to $p^*_i= (A^m-A^{-m}) \hbar m /  m \zeta(\hbar m )  $. Denote the coefficients of expansion of $z/\zeta(z)$ by $B_i$, that is, $z/\zeta(z) = \sum_{i=0}^\infty B_i z^{2i}$. Then $s_{m,a} = B_a m^{2a-1} (A^m-A^{-m}) $. Therefore, Equation~\eqref{eq:cutjoin} in our case takes the following form:

\begin{lemma} We have:
\begin{align}\label{eq:cutjoinOV}
&\left(\sum_{i=1}^n D_{X_i}\right) \tilde{H}_{g,n}\\ \nonumber
&= \sum_{\substack{m \geq 1 \\ a\geq 0 \\ d \geq 0}} \sum_{\substack{t \geq 0 \\ 1 \leq k \leq n}} B_a m^{2a} \dfrac{A^m-A^{-m}}{m} \left(\dfrac{P}{Q}m\right)^{2d+t}[u^{2d}] \dfrac{u}{\zeta(u)} \dfrac{\zeta(u D_{X_k})}{uD_{X_k}} X_k^m 
\;\widetilde{\mathcal{D}\mathcal{H}}^{k,t}_{g-d-a,n},
\end{align}
where
\begin{equation}\label{eq:Hcurlytilde}
\widetilde{\mathcal{D}\mathcal{H}}^{k,t}_{g-d-a,n} :=\sum_{\substack{\ell = 1,\dots, t \\ \text{or } \ell=0 \text{ if } t=0} }\dfrac{1}{\ell!}\sum_{\substack{
		\bigsqcup_{j=1}^\ell K_j = \llbracket n \rrbracket \setminus \{k\} \\ 
		\bigsqcup_{j=1}^\ell T_j =  \llbracket t \rrbracket\\ 
		T_j \neq \emptyset \\
		g-d-a = \sum_{j=1}^\ell g_j + t - \ell \\
		g_1,\ldots,g_{\ell} \geq 0 
}}
\prod_{i=1}^t \dfrac{1}{t!}\left(\bigg\rfloor_{\xi_i=X_k} \dfrac{\zeta(uD_{\xi_i})}{u} \right) 
\bigg[\prod_{j = 1}^{\ell} \tilde{H}_{g_j,|T_j|+|K_j|}(\xi_{T_j},X_{K_j})\bigg]
\end{equation}
\end{lemma}

\begin{remark}
	Note that the sum over $a$ and $d$ in \eqref{eq:cutjoinOV} is finite due to $\widetilde{\mathcal{D}\mathcal{H}}^k_{g-d-a,n}$ vanishing for $d+a>g$.
\end{remark}

According to~\cite[Theorem 8.1]{DPSS}, the formal power series $\tilde H_{g,n}$ can also be considered as the expansions of the globally defined meromorphic functions on the rational spectral curve with the global coordinate $z$, where $X=X(z)$ is given in Equation~\eqref{eq:SpectralCurve-Definition-X-y}. Abusing the notation a little bit, we denote these meromorphic functions also by $\tilde H_{g,n}$. Our goal is to interpret Equation~\eqref{eq:cutjoinOV} as an equation for meromorphic functions on the curve.

Let $X_i = X(z_i)$. Assume that $|A| \leq |A^{-1}|$. We have the following lemma:

\begin{lemma} For $(g,n)\neq (0,1)$, for each $k=1,\dots,n$ the corresponding summand on the right hand side of Equation~\eqref{eq:cutjoinOV} is a finite linear combination of rational functions in $z_1,\dots,z_n$ with the coefficients given by power series in $z_k$ that converge to  holomorphic functions on the disk $\{|z_k|<|A|\}$. 
\end{lemma}

\begin{proof} Note that the terms in the sum in \eqref{eq:Hcurlytilde} sometimes contain factors $\tilde{H}_{0,1}(\xi_{s})$ (this happens when one of the $g_j$ is zero, and the corresponding $K_j$ is empty, while $|T_j|=1$). Now let us define by 
\begin{equation}
\widecheck{\mathcal{D}\mathcal{H}}^{k,t}_{g-d-a,n}
\end{equation}
the sum \eqref{eq:Hcurlytilde} where in all the terms which contain factors of $\tilde{H}_{0,1}(\xi_{s})$ the corresponding operators $u^{-1}\zeta(uD_{\xi_s})$ are replaced by $u^{-1}\zeta(uD_{\xi_s})- D_{\xi_s}$. In other words, $\widecheck{\mathcal{D}\mathcal{H}}^{k,t}_{g-d-a,n}$ is $\widetilde{\mathcal{D}\mathcal{H}}^{k,t}_{g-d-a,n}$ where we represent $u^{-1}\zeta(uD_{\xi_i})$ as $D_{\xi_i}+\mathcal{O}_i$, expand all these brackets, and then drop all terms containing $D_{\xi_i}\tilde{H}_{0,1}(\xi_{i})$. Note that this definition implies that any sum of the form
\begin{equation}\label{eq:DHorigkappa}
\sum_{t=0}^\infty \varkappa_t \widetilde{\mathcal{D}\mathcal{H}}^{k,t}_{g-d-a,n}
\end{equation}
can be rewritten as
\begin{equation}\label{eq:checkDHsum}
\sum_{t_1,t_2=0}^\infty \varkappa_{t_1+t_2}\dfrac{1}{t_1!} \left(D_{X_k}\tilde{H}_{0,1}(X_k)\right)^{t_1} \widecheck{\mathcal{D}\mathcal{H}}^{k,t_2}_{g-d-a,n}\,.
\end{equation}
by setting $t=t_1+t_2$, where $t_1$ counts the $D_{\xi_i}\tilde{H}_{0,1}(\xi_{i})$ factors. Here we used the fact that 
\begin{equation}
\prod_{i=1}^{t_1}\bigg\rfloor_{\xi_{s_i}=X_k} D_{\xi_{s_i}} \tilde{H}_{0,1}(\xi_{s_i}) = \left(D_{X_k}\tilde{H}_{0,1}(X_k)\right)^{t_1}.
\end{equation}

Apply this \eqref{eq:DHorigkappa} $\rightsquigarrow$ 
\eqref{eq:checkDHsum} procedure to the $k$-th summand on the right hand side of Equation~\eqref{eq:cutjoinOV}. We obtain:
\begin{align}
\label{eq:cutjoinOVrearr11}
\sum_{\substack{m \geq 1 \\ a\geq 0 \\ d \geq 0}}  B_a m^{2a} \dfrac{A^m-A^{-m}}{m} \sum_{t_1,t_2 = 0}^\infty\left(\dfrac{P}{Q}m\right)^{2d+t_1+t_2}[u^{2d}] \dfrac{u}{\zeta(u)} \dfrac{\zeta(u D_{X_k})}{uD_{X_k}} X_k^m\; \dfrac{1}{t_1!} \left(D_{X_k}\tilde{H}_{0,1}(X_k)\right)^{t_1} \widecheck{\mathcal{D}\mathcal{H}}^{k,t_2}_{g-d-a,n}.
\end{align}
Note that only finitely many terms in the sum over $t_2$ give a nonvanishing contribution to the RHS of \eqref{eq:cutjoinOVrearr11}.
This can be seen as follows. $\widecheck{\mathcal{D}\mathcal{H}}^{k,t_2}_{g-d-a,n}$ is still a sum of terms of the form
\begin{equation}
\prod_{j = 1}^{\ell} \tilde{H}_{g_j,|T_j|+|K_j|}(\xi_{T_j},X_{K_j}),
\end{equation}
acted upon with some operators, and subject to the conditions $\bigsqcup_{j=1}^\ell K_j = \llbracket n \rrbracket \setminus \{k\},\, 
\bigsqcup_{j=1}^\ell T_j =  \llbracket t_2 \rrbracket,\, 
T_j \neq \emptyset,\,
g-d-a = \sum_{j=1}^\ell g_j + t_2 - \ell,\,
g_1,\ldots,g_{\ell} \geq 0 $.
Now for each $\tilde{H}_{g_j,|T_j|+|K_j|}(\xi_{T_j},X_{K_j})$ entering the expression for $\widecheck{\mathcal{D}\mathcal{H}}^{k,t_2}_{g-d-a,n}$ either $K_j\neq \emptyset$, or $g_i\neq 0$, or $|T_j|\geq 2$, or it is actually a $\tilde H_{0,1}$. All cases but the latter correspond to increases in $\sum_{j=1}^\ell g_j + t_2 - \ell$, and it is bounded from above by $g$. In the latter case the operator acting on $\tilde H_{0,1}$ is $u^{-1}\zeta(uD_{\xi_s}) - D_{\xi_s}$, whose series expansion in $u$ starts with a term proportional to $u^1$, and the total power of $u$ is bounded from above by $2 d$ 
(which is itself bounded from above by $2g$), due to $[u^{2d}]$ being present in the formula. Thus the total number of $\xi$-arguments of $\widecheck{\mathcal{D}\mathcal{H}}^{k,t_2}_{g-d-a,n}$, denoted by $t_2$, is bounded from above for fixed $g$ and $n$.

We can take the sum over $t_1$ in~\eqref{eq:cutjoinOVrearr11} and obtain
\begin{align}
\label{eq:cutjoinOVrearr1}
&\sum_{\substack{m \geq 1 \\ a\geq 0 \\ d \geq 0}} B_a m^{2a} \dfrac{A^m-A^{-m}}{m} \sum_{t_2=0}^\infty\left(\dfrac{P}{Q}m\right)^{2d+t_2}[u^{2d}] \dfrac{u}{\zeta(u)} \dfrac{\zeta(u D_{X_k})}{uD_{X_k}} X_k^m\; \exp\left(\dfrac{P}{Q}\,m\,D_{X_k}\tilde{H}_{0,1}(X_k)\right) \widecheck{\mathcal{D}\mathcal{H}}^{k,t_2}_{g-d-a,n}\,.
\end{align}
Recall that our spectral curve has the form
\begin{equation}
X(z) = z \exp\left(-\dfrac{P}{Q} y(z)\right)\qquad \text{and} \qquad y=D_{X}\tilde{H}_{0,1}(X),
\end{equation}
which means that 
\begin{equation}
X_k^m\; \exp\left(\dfrac{P}{Q}\,m\,D_{X_k}\tilde{H}_{0,1}(X_k)\right) = z_k^m,
\end{equation}
and thus 
we can rewrite \eqref{eq:cutjoinOVrearr1} in terms of variables $z_k$ instead of $X_k$ as
\begin{align}\label{eq:cutjfexpr}
&\sum_{\substack{m \geq 1 \\ a\geq 0 \\ d \geq 0}} B_a m^{2a} \dfrac{A^m-A^{-m}}{m} \sum_{t_2=0}^\infty\left(\dfrac{P}{Q}m\right)^{2d+t_2}[u^{2d}] \dfrac{u}{\zeta(u)} \dfrac{\zeta(u 
	D_{X_k}
	)}{u
	D_{X_k}
} z_k^m \widecheck{\mathcal{D}\mathcal{H}}^{k,t_2}_{g-d-a,n}\\ \nonumber
&= \sum_{\substack{a\geq 0 \\ d \geq 0}}   \sum_{t_2=0}^\infty B_a\left(\dfrac{P}{Q}\right)^{2d+t_2}[u^{2d}] \dfrac{u}{\zeta(u)} \dfrac{\zeta(u
	D_{X_k}
	)}{u
	D_{X_k}
} \left(\widecheck{\mathcal{D}\mathcal{H}}^{k,t_2}_{g-d-a,n} \sum_{m=1}^\infty m^{2a+2d+t_2-1} \left(A^m-A^{-m}\right)z_k^m\right)
\end{align}
Consider the sum over $m$ in the above expression. It can be rewritten as follows:
\begin{equation}
\label{eq:msum}
\sum_{m=1}^\infty m^{2a+2d+t_2-1} \left(A^m-A^{-m}\right)z_k^m = \left(
z\dfrac{d}{dz}
\right)^{2a+2d+t_2}\sum_{m=1}^\infty \dfrac{A^m-A^{-m}}{m}z_k^m.
\end{equation}
Under the assumption that $|z_k|<|A| < |A^{-1}|$ the sum 
\begin{equation}
\sum_{m=1}^\infty \dfrac{A^m-A^{-m}}{m}z_k^m
\end{equation}
converges to 
\begin{equation}\label{eq:logs}
\log\left(1-A^{-1}z_k\right)-\log\left(1-A_kz\right).
\end{equation}
Note that there is a nonzero contribution of $2a+2d+t_2 = 0$ only for $(g,n)=(0,1)$. Indeed, for $a=d=t_2=0$ the factor $\widecheck{\mathcal{D}\mathcal{H}}^{k,t_2}_{g-d-a,n}$ is nonzero only for $g=0$, $n=1$, since otherwise in the corresponding \eqref{eq:Hcurlytilde}-like expression either the condition $g-d-a = \sum_{j=1}^\ell g_j + t - \ell$ or the condition $\bigsqcup_{j=1}^\ell K_j = \llbracket n \rrbracket \setminus \{k\}$, respectively, cannot be satisfied, which makes the sum empty.

Thus for $(g,n)\neq (0,1)$ we can assume that $2a+2d+t_2 > 0$. This means that function \eqref{eq:logs} once substituted into \eqref{eq:cutjfexpr} via \eqref{eq:msum} gets acted upon by at least one instance of the operator $z d/dz$, which makes it a rational function in $z_k$. 
All the sums are now finite, and all the other factors and operators preserve the property of a function of $z_k$ being a rational function, which implies the statement of the lemma (taking also into account the fact that the dependence of our $k$-th summand on all other $z_i$ for $i\neq k$ is rational, as follows from all $H_{g,n}$, $(g,n)\neq(0,1)$, being rational functions on the spectral curve, as implied by~\cite[Theorem 8.1]{DPSS}).
\end{proof}

An immediate corollary of this lemma is the following:

\begin{corollary} For $|z_i|<|A|$, $i=1,\dots,n$, the cut-and-join equation~\eqref{eq:cutjoinOV} implies the following identity of meromorphic functions:
\begin{align}
\label{eq:cutjoinOVrearr2}
&\left(\sum_{i=1}^n D_{X_i}\right) \tilde{H}_{g,n}\\ \nonumber
&= \sum_{\substack{a\geq 0 \\ d \geq 0}} \sum_{1 \leq k \leq n}   \sum_{t_2=0}^\infty B_a\left(\dfrac{P}{Q}\right)^{2d+t_2}[u^{2d}] \dfrac{u}{\zeta(u)} \dfrac{\zeta(u
	D_{X_k}
	)}{u
	D_{X_k}
} \left(\widecheck{\mathcal{D}\mathcal{H}}^{k,t_2}_{g-d-a,n} \sum_{m=1}^\infty m^{2a+2d+t_2-1} \left(A^m-A^{-m}\right)z_k^m\right)
\\ \nonumber
&= \sum_{\substack{a\geq 0 \\ d \geq 0}} \sum_{1 \leq k \leq n}   \sum_{t_2=0}^\infty B_a\left(\dfrac{P}{Q}\right)^{2d+t_2}[u^{2d}] \dfrac{u}{\zeta(u)} \dfrac{\zeta(u
	D_{X_k}
	)}{u
	D_{X_k}
} \left(\widecheck{\mathcal{D}\mathcal{H}}^{k,t_2}_{g-d-a,n} \left(
z_k\dfrac{d}{dz_k}
\right)^{2a+2d+t_2} \log \dfrac{1-A^{-1}z_k}{1-Az_k}  \right)\,.
\end{align}
\end{corollary}

\subsection{Analytic continuation of the cut-and-join equation} \label{sec:AnalyticCont}

For the purposes of topological recursion we are interested of the behavior of the expression \eqref{eq:cutjoinOVrearr2}
on the whole spectral curve (specifically, near the branching points of $X(z)$).  
We have to analytically continue the functions from the previous subsection beyond the disks $\{|z|<|A|\}$.

Recall that we assume $|A| < |A^{-1}|$, and there are three rings in which we would like to consider the re-expansion of the right hand side of~\eqref{eq:logs}:
\begin{align}\label{eq:R1def}
	R_1 &\coloneqq \{0<|z|<|A|\}, \\
	R_2 &\coloneqq \{|A|<|z|<|A|^{-1}\}, \\ \label{eq:R3def}
	R_3 &\coloneqq \{|A|^{-1}<|z|<\infty\}
\end{align}
($R_1$ is just the original disk $\{|z|<|A|\}$, but from now on we exclude $z=0$ and call it a ring to unify the terminology).

We have the following re-expansions of  \eqref{eq:logs} in these rings:
\begin{align}  \label{eq:LogsReexpansions}
& z\in R_1: & & \sum_{m=1}^\infty \dfrac{A^m-A^{-m}}{m}z_k^m, \\ \notag
& z\in R_2: & & -\log z +  \sum_{m=1}^\infty \dfrac{A^{-m}}{m} z^{-m} - \sum_{m=1}^\infty \dfrac{A^{-m}}{m} z^m, \\ \notag
& z\in R_3: & &\sum_{m=1}^\infty \dfrac{A^{-m}-A^{m}}{m} z^{-m} . 
\end{align}

We have
\begin{lemma} Consider the point $(q_1,\dots,q_n)$, where $q_i\in R_{t_i}$, $i=1,\dots,n$, $t_i\in \{1,2,3\}$. Near this point, the analytic extension of Equation~\eqref{eq:cutjoinOVrearr2} is still the sum of $n$ summands for $k=1,\dots,n$, where the $k$-th summand is 
\begin{align} \label{eq:RHD-CJ-Ring1}
\sum_{\substack{a\geq 0 \\ d \geq 0}}   \sum_{t_2=0}^\infty B_a\left(\dfrac{P}{Q}\right)^{2d+t_2}[u^{2d}] \dfrac{u}{\zeta(u)} \dfrac{\zeta(u
	D_{X_k}
	)}{u
	D_{X_k}
} \left(\widecheck{\mathcal{D}\mathcal{H}}^{k,t_2}_{g-d-a,n} 
\sum_{m=1}^\infty m^{2a+2d+t_2-1} (A^m-A^{-m})z_k^m \right),  \\ \notag
q_k\in R_1; \\ \label{eq:RHD-CJ-Ring2}
\sum_{\substack{a\geq 0 \\ d \geq 0}}   \sum_{t_2=0}^\infty B_a\left(\dfrac{P}{Q}\right)^{2d+t_2}[u^{2d}] \dfrac{u}{\zeta(u)} \dfrac{\zeta(u
	D_{X_k}
	)}{u
	D_{X_k}
} \Bigg(\widecheck{\mathcal{D}\mathcal{H}}^{k,t_2}_{g-d-a,n} \Bigg(-\delta_{2a+2d+t_2,0} \log z_k -\delta_{2a+2d+t_2,1}
\\ \notag 
-\sum_{m=1}^\infty (-m)^{2a+2d+t_2-1} A^{-m} z_k^{-m}
-\sum_{m=1}^\infty m^{2a+2d+t_2-1} A^{-m}z_k^m \Bigg)\Bigg),  
\\ \notag
q_k\in R_2; \\ \label{eq:RHD-CJ-Ring3}
\sum_{\substack{a\geq 0 \\ d \geq 0}}   \sum_{t_2=0}^\infty B_a\left(\dfrac{P}{Q}\right)^{2d+t_2}[u^{2d}] \dfrac{u}{\zeta(u)} \dfrac{\zeta(u
	D_{X_k}
	)}{u
	D_{X_k}
} \left(\widecheck{\mathcal{D}\mathcal{H}}^{k,t_2}_{g-d-a,n} 
\sum_{m=1}^\infty (-m)^{2a+2d+t_2-1} (A^m-A^{-m})z_k^{-m} \right),  \\ \notag
q_k\in R_3 
\end{align}
\end{lemma}

\begin{proof} Observe that the choice of the ring where the point $q_k$ belongs affects only the $k$-th summand on the right hand side of~\eqref{eq:cutjoinOVrearr2} (all other summands depend on $z_k$ through the globally defined rational functions). To this end, we have to replace the function $ \log\left(1-A^{-1}z_k\right)-\log\left(1-Az_k\right)$ in the $k$-th summand on the right hand side of Equation~\eqref{eq:cutjoinOVrearr2} by one of its expansions given by~\eqref{eq:LogsReexpansions}.
\end{proof}

\begin{remark}
	The functions $z_k$ in the expansions of $\log\left(1-A^{-1}z_k\right)-\log\left(1-Az_k\right)$ in~\eqref{eq:RHD-CJ-Ring1}-\eqref{eq:RHD-CJ-Ring3} can now be replaced back by $X_k \exp\Big(\frac PQ D_{X_k} \tilde H_{0,1}(X_k)\Big)$, $k=1,\dots,n$. 
\end{remark}

\begin{remark}
Note that for $(g,n)=(0,1)$ (which is the only case when $2a+2d+t_2=0$ gives a nonzero contribution) the cut-and-join equation just turns into the equation of the spectral curve.
\end{remark} 


%
%

\section{Quadratic loop equations} \label{sec:QLE}

In this Section we use the analytic extension of the cut-and-join equation, obtained in the previous section in order to prove Theorem~\ref{thm:QLE}, that is, to prove the quadratic loop equation for the $n$-point functions. The latter is the only missing ingredient in the proof of the topological recursion outlined in Section~\ref{sec:ProofOfTR}. Essentially, we follow the same scheme of proof as in~\cite{KPS}, and it is based on the system of formal corollaries of quadratic loop equations obtained in~\cite{DKPS}.

\subsection{Notation and conventions} Let $p$ be a critical point of $X$. We denote by $S_zf(z)$ (respectively, $\Delta_zf(z)$) the symmetrization $f(z)+f(\sigma(z))$ (respectively, the antisymmetrization $f(z)-f(\sigma(z))$) of a function $f(z)$ with respect to the deck transformation $\sigma$ of $X(z)$ near $p$.

One of the properties of the operators $S_z$ and $\Delta_z$ that we use below is the identity
\begin{equation}\label{eq:SonDiagonal}
S_zf(z_1,\dots,z_n)\Big|_{z_i = z} = 2^{1-n} \sum_{\substack{I \sqcup J = \llbracket n \rrbracket \\ |J| \in 2\mathbb{Z}}} \Big( \prod_{i \in I} \cS_{z_i} \Big) \Big( \prod_{j \in J} \Delta_{z_j} \Big)f(z_1,\dotsc, z_n) \Big|_{z_i = z}
\end{equation}
that is valid for any $n\geq 1$~\cite{BKLPS,DKPS,KPS}.

 Furthermore, introduce the notation 
\begin{align}\label{eq:defW}
& W_{g,t+n} (\xi_{\llbracket t \rrbracket} , X_{\llbracket n \rrbracket})\coloneqq D_{\xi_1}\cdots D_{\xi_m}D_{X_1}\cdots D_{X_n} \tilde{H}_{g,t+n}(\xi_{\llbracket t \rrbracket},X_{\llbracket n \rrbracket}); \\ \label{eq:defcW}
& \cW_{g,t,n}(w_{\llbracket t \rrbracket} \mid z_{\llbracket n \rrbracket}) \coloneqq 
\sum_{\ell =1}^t 
\frac{1}{\ell!}
\!\!
\sum_{\substack{
		\bigsqcup_{j=1}^\ell K_j = \llbracket n \rrbracket \\ 
		\bigsqcup_{j=1}^\ell T_j =  \llbracket t \rrbracket\\ 
		T_j \neq \emptyset \\
		g= \sum_{j=1}^\ell g_j + t - \ell \\
		g_1,\ldots,g_{\ell} \geq 0 
}}  \!\!
\prod_{j = 1}^{\ell} W_{g_j,|M_j|+|K_j|}(\xi_{M_j},x_{K_j}),
\end{align}
where we assume that $\xi_i\coloneqq X(w_i)$, $i=1,\dots,m$, and $X_j\coloneqq X(z_j)$, $j=1,\dots,n$. Denote also~(cf. \eqref{Qrdef})
\begin{align} 
\sum_{d \geq 0} \tilde Q_{d,m,t}(z_0)z^{2d} & = \frac{z}{\zeta(z)} \frac{\zeta(zD_{X(z_0)})}{zD_{X(z_0)}} \circ X(z_0)^m\circ
\prod_{j = 1}^t \bigg( \left[\vert_{w_j = z_0}\right] \circ \frac{\zeta(zD_{X(w_j)})}{zD_{X(w_j)}}\bigg),
\end{align}
where
\begin{align}
\left[\vert_{w_j = z_0}\right] F(w) &\coloneq \Res_{w=z} F(w)\frac{dX(w)}{X(w)-X(z)}
\,.
\end{align}
\begin{remark}\label{rem:QQtilde}
	Note that the definition of the $\tilde Q_{d,m,t}$ differs from the definition of $Q_{d,m,t}$ given by  Equation~\eqref{Qrdef} by an extra factor of $D_{X(w_j)}$ in the denominator in the fraction in the brackets. This change is related to the fact that in this section we prefer to work with $\cW$ rather than $\mathcal{D} \tilde{\mathcal{H}}$.
\end{remark}

Introduce one more notation:
\begin{equation}
QLE_{g,n+1}=QLE_{g,n+1}(z_0;z_{\llbracket n \rrbracket})\coloneqq \left[\vert_{w_1 = z_0}\right] \left[\vert_{w_2 = z_0}\right] \Delta_{w_1} \Delta_{w_2} \cW_{g,2,n}(w_1,w_2 \mid z_{\llbracket n \rrbracket}).
\end{equation}
Since $W_{g,n}$ satisfy the linear loop equations for any $g\geq 0$, $n\geq 1$ \cite[Theorem 8.1 and Remark 8.2]{DPSS}, Theorem~\ref{thm:QLE} can be equivalently reformulated as follows (cf.~\cite[Section 2.4]{BS} and~\cite[Lemma 20]{BKLPS}, and~\cite[Section 3.2]{DKPS}):

\begin{theorem} \label{thm:QLE-equiv} For any $g\geq 0$ and $n\geq 0$ the expression $QLE_{g,n+1}$ is holomorphic in $z_0$ at $z_0\to p$.
\end{theorem}

\begin{remark} Note that in order to prove the equivalence of the statements of Theorems~\ref{thm:QLE-equiv} and~\ref{thm:QLE} one has to adjust a little bit the definition of $\Delta_{w_1} \Delta_{w_2} W_{0,2}(w_1,w_2)$ to treat properly the pole at the anti-diagonal. We refer for the details to~\cite[Section 3.1]{DKPS}. This adjustment doesn't affect the argument below.
\end{remark}

\subsection{Guide to the proof} Let us summarize the main steps of the proof of Theorem~\ref{thm:QLE-equiv} presented below. Firstly, we use the linear loop equations proved in~\cite{DPSS} and the analytic continuation of the cut-and-join equation obtained in Section~\ref{sec:AnalyticCont} in order to obtain some expressions that depend on $g$ and $n$ and are holomorphic in $z_0$ near the point $p$. These expressions depend on the ring \eqref{eq:R1def}-\eqref{eq:R3def} where $p$ belongs and are collected in Lemma~\ref{lem:holomorphicexpressions}.

Secondly, we assume that the quadratic loop equation is satisfied for any $(g',n')$ such that $2g'-2+(n'+1) < 2g-2+(n+1)$ and use this assumption along with the linear loop equations in order to get simpler expressions than in Lemma~\ref{lem:holomorphicexpressions} that are still holomorphic in $z_0$ near the point $p$. The results of these computations are collected in Corollary~\ref{cor:holomorphic}.

The expressions in Corollary~\ref{cor:holomorphic} are represented as $QLE_{g,n+1}$ times some prefactors defined locally in a neighborhood of the point $p$ (the form of the prefactor depends on the ring where the point $p$ belongs). As a third step of the proof, we prove in Lemma~\ref{lem:coeffinvert} that these prefactors are invertible holomorphic functions. 
This observation allows us to prove Theorem~\ref{thm:QLE-equiv} by induction, see Section~\ref{sec:ProofQLE} .

\subsection{Holomorphic expressions} The goal of this Section is to obtain some expressions in terms of the functions $\cW$ and operators $\tilde Q$ that are holomorphic in a distinguished variable $z_0$ at $z_0\to p$.

\begin{lemma} \label{lem:holomorphicexpressions}
Depending on the ring \eqref{eq:R1def}-\eqref{eq:R3def} where $p$ belongs, one of the following expressions is holomorphic in $z_0$ near $p$:
\begin{align} \label{eq:BasicSymmetricExpression-R1}	
\cS_{z_0} \left[\sum_{\substack{m\geq 1\\a\geq 0}} 
\sum_{\substack{t \geq 1 \\ d\geq 0 }} (A^m-A^{-m})B_a 
\frac{(\frac PQ)^{t+2d} m^{t-1+2d+2a}}{t!} \tilde Q_{d,m,t}(z_0) \cW_{g-d-a,t,n} (w_{\llbracket t \rrbracket} \mid z_{\llbracket n \rrbracket}) \right],  \\ \notag
\text{if } p\in R_1; \\ \label{eq:BasicSymmetricExpression-R2}
\cS_{z_0} \left[\sum_{\substack{m\geq 1\\a\geq 0}} 
\sum_{\substack{t \geq 1 \\ d\geq 0 }} A^{-m}B_a 
\frac{(\frac PQ)^{t+2d} m^{t-1+2d+2a}}{t!} \left((-1)^t \tilde Q_{d,-m,t}(z_0) - \tilde Q_{d,m,t}(z_0)\right) 
\cW_{g-d-a,t,n} (w_{\llbracket t \rrbracket} \mid z_{\llbracket n \rrbracket}) \right]  , \\ \notag
\text{if } p\in R_2; \\ \label{eq:BasicSymmetricExpression-R3}
S_{z_0} \left[\sum_{\substack{m\geq 1\\a\geq 0}} 
\sum_{\substack{t \geq 1 \\ d\geq 0 }} (A^{-m}-A^m)B_a 
\frac{(\frac PQ)^{t+2d} m^{t-1+2d+2a}(-1)^t }{t!} \tilde Q_{d,-m,t}(z_0) \cW_{g-d-a,t,n} (w_{\llbracket t \rrbracket} \mid z_{\llbracket n \rrbracket}) \right]  , \\ \notag
\text{if } p\in R_3.
\end{align}
\end{lemma}

\begin{proof} Consider Equation~\eqref{eq:cutjoinOVrearr2} assuming that we have $n+1$ variables, $z_0,\dots,z_n$, that is, its left hand side reads $\Big(\sum_{i=0}^n D_{X_i} \Big)\tilde H_{g,n+1}$. The summation over $k$ on the right hand side runs also from $0$ to $n$, and near the point $(q_0,\dots,q_n)\in R_{i_0}\times \cdots \times R_{i_n}$, $1\leq i_0,\dots,i_n\leq 3$, the right hand side turns into a sum of $n+1$ summands given by Formulas \eqref{eq:RHD-CJ-Ring1}, \eqref{eq:RHD-CJ-Ring2}, or \eqref{eq:RHD-CJ-Ring3} (where one has still to change index $n$ in the factor $\widecheck{\mathcal{D}\mathcal{H}}^{k,t_2}_{g-d-a,n}$ to $n+1$ to reflect that we have $n+1$ variables now), depending on the ring $R_{i_j}$ where the point $q_j$ belongs, $j=0,\dots,n$. We have:
	
\begin{equation} \label{eq:Cut-And-Join-In-Rings}
\Big(\sum_{i=0}^n D_{X_i} \Big)\tilde H_{g,n+1} = \sum_{k=0}^n Y_{k,i_k} \qquad \text{in} \qquad R_{i_0}\times R_{i_1}\times \cdots \times R_{i_n}, 
\end{equation}
where $Y_{k,i_k}$, $k=0,\dots,n$, $i_k=1,2,3$, is equal to 
\begin{align} \label{eq:RHD-CJ-Ring1-n+1}
\sum_{\substack{a\geq 0 \\ d \geq 0}}   \sum_{t_2=0}^\infty B_a\left(\dfrac{P}{Q}\right)^{2d+t_2}[u^{2d}] \dfrac{u}{\zeta(u)} \dfrac{\zeta(u
	D_{X_k}
	)}{u
	D_{X_k}
} \left(\widecheck{\mathcal{D}\mathcal{H}}^{k,t_2}_{g-d-a,n+1} 
\sum_{m=1}^\infty m^{2a+2d+t_2-1} (A^m-A^{-m})z_k^m \right)  \\ \notag
\text{for}\ i_k=1; \\ 
\label{eq:RHD-CJ-Ring2-n+1}
\sum_{\substack{a\geq 0 \\ d \geq 0}}   \sum_{t_2=0}^\infty B_a\left(\dfrac{P}{Q}\right)^{2d+t_2}[u^{2d}] \dfrac{u}{\zeta(u)} \dfrac{\zeta(u
	D_{X_k}
	)}{u
	D_{X_k}
} \Bigg(\widecheck{\mathcal{D}\mathcal{H}}^{k,t_2}_{g-d-a,n+1} \Bigg(-\delta_{2a+2d+t_2,0} \log z_k -\delta_{2a+2d+t_2,1}
\\ \notag 
-\sum_{m=1}^\infty (-m)^{2a+2d+t_2-1} A^{-m} z_k^{-m}
-\sum_{m=1}^\infty m^{2a+2d+t_2-1} A^{-m}z_k^m \Bigg)\Bigg)
\\ \notag
\text{for}\ i_k=2; 
\\ \label{eq:RHD-CJ-Ring3-n+1}
\sum_{\substack{a\geq 0 \\ d \geq 0}}   \sum_{t_2=0}^\infty B_a\left(\dfrac{P}{Q}\right)^{2d+t_2}[u^{2d}] \dfrac{u}{\zeta(u)} \dfrac{\zeta(u
	D_{X_k}
	)}{u
	D_{X_k}
} \left(\widecheck{\mathcal{D}\mathcal{H}}^{k,t_2}_{g-d-a,n+1} 
\sum_{m=1}^\infty (-m)^{2a+2d+t_2-1} (A^m-A^{-m})z_k^{-m} \right)  \\ \notag
\text{for}\ i_k=3. 
\end{align}

Let $p$ be a critical point of $X$ and assume that $R_{i_0}\ni p$. Apply $S_{z_0}$ (the symmetrization with respect to the deck transformation near the point $p$ in the coordinate $z_0$) to the both sides of Equation~\eqref{eq:Cut-And-Join-In-Rings}. Note that the linear loop equations proved in~\cite{DPSS} imply that $S_{z_0} \Big(\sum_{i=0}^n D_{X_i} \Big)\tilde H_{g,n+1}$ and $S_{z_0}Y_{k,i_k}$, $k=1,\dots,n$ are holomorphic for $z_0\to p$ (in the latter cases the dependence on $z_0$ is only present in the factor $\widecheck{\mathcal{D}\mathcal{H}}^{k,t_2}_{g-d-a,n+1}$). Therefore, we obtain that $S_{z_0} Y_{0,i_0}$ is holomorphic for $z_0\to p$ if $p\in R_{i_0}$, $i_0=1,2,3$.  

Consider the expression $S_{z_0} Y_{0,i_0}$. If $i_0=2$, there are two exceptional terms, one for $2a+2d+t_2=0$ and one for $2a+2d+t_2=1$. 
These exceptional terms occur only when $a=d=0$ and either $t_2=0$  or $t_2=1$. In the first case, Equation~\eqref{eq:Cut-And-Join-In-Rings} is just the equation of the spectral curve, and $t_1=1$. In the second case, $t_1=0$. Therefore, in both exceptional cases the corresponding exceptional terms in $S_{z_0} Y_{0,i_0}$ are holomorphic by the linear loop equations.

For any $i_0=1,2,3$, consider all non-exceptional terms in $S_{z_0} Y_{0,i_0}$. In each case we have there a factor, which is the sum of two terms of the shape $\pm \sum_{m=1}^\infty (\pm m)^{2a+2d+t_2-1} A^{\pm m} z_0^{\pm m}$. Replace $z_0$ by $X_0 \exp\Big(\frac PQ D_{X_0} \tilde H_{0,1}(X_0)\Big)$. Then the computation in Section~\ref{sec:ChangeFromQW-to-CheckH} run backwards explains how to assemble the expansions of the non-exceptional terms of $S_{z_0} Y_{0,i_0}$ in $D_{X_0} \tilde H_{0,1}(X_0)$ for $i_0=1,2,3$, respectively, into the formulas~\eqref{eq:BasicSymmetricExpression-R1}, \eqref{eq:BasicSymmetricExpression-R2}, \eqref{eq:BasicSymmetricExpression-R3}, respectively. 
More precisely, we apply all the steps from \eqref{eq:cutjoinOVrearr2} back to \eqref{eq:cutjoinOV}-\eqref{eq:Hcurlytilde} to $D_{X_1}\cdots D_{X_n}Y_{0,i_0}$, $i_0=1,2,3$, and take into account the differences between $\widetilde{\mathcal{D}\mathcal{H}}$, $Q$  and $\cW$, $\tilde Q$ mentioned in Remark \ref{rem:QQtilde}. Note that the $t_2=0$ term gets canceled after the application of $S_{z_0}$, and thus the sum over $t$ in~~\eqref{eq:BasicSymmetricExpression-R1}-\eqref{eq:BasicSymmetricExpression-R3} starts with $t=1$.
\end{proof}


\subsection{Obstructions to holomorphicity}
Assume that the quadratic loop equations hold for all $(g',n')$ with $2g'-2+(n'+1)<2g-2+(n+1)$. Then \cite[Lemma 4.2 and Lemma 4.3]{KPS} (itself direct corollaries of \cite[Corollary 3.4 and Remark 3.3]{DKPS}) say:
\begin{lemma}\label{lem:secondRSpinLemma}
For any $r \geq 0$ and any $h, k \geq 0$ such that $ 2h-1+k -r\leq 2g-2+n$, the expression
\begin{equation}\label{eq:corollaryDKPS}
\sum_{t=1}^{r+1}\frac{1}{t!} \sum_{\substack{2\alpha_1 + \dotsb + 2\alpha_t \\ + t = r+1}} \prod_{j=1}^t \bigg( \left[\vert_{w_j = z_0}\right] \frac{D_{X(w_j)}^{2\alpha_j}}{(2\alpha_j+1)!} \bigg) \sum_{\substack{I\sqcup J = \llbracket t \rrbracket \\ |I|\in 2\mathbb{Z} }} \prod_{i\in I} \Delta_{w_i} \prod_{j\in J} \cS_{w_j} \cW_{h+(t-r-1)/2,t,k} (w_{\llbracket t\rrbracket} \mid z_{\llbracket k\rrbracket}) 
\end{equation}
is holomorphic in $z_0$ at $ z_0 \to p$. Moreover, if we apply to it an arbitrary operator $O(X(z_0), D_{X(z_0)})$ polynomial in $X(z_0)$ and $D_{X(z_0)}$ this expression remains holomorphic in $z_0$ at $ z_0 \to p$. Furthermore, for any $r\geq 1$, the expression
	\begin{align}\label{eq:expr-r}
	& \sum_{\substack{k,\alpha_1,\dots,\alpha_{2k}\\ \ell, \beta_1,\dots,\beta_\ell \\ 2k+2\alpha_1+\cdots+2\alpha_{2k} \\ +\ell + 2\beta_1+\cdots+2\beta_{\ell} = r+1}}\!\!\!\!\!\!\! \frac{1}{\ell! (2k)!}\prod_{i=1}^\ell \left[\vert_{w'_i = z_0}\right] \frac{D_{X(w'_i)}^{2\beta_i}}{(2\beta_i+1)!} \cS_{w'_i} 
	\prod_{i=1}^{2k} \left[\vert_{w_i = z_0}\right] \frac{D_{X(w_i)}^{2\alpha_i}}{(2\alpha_i+1)!} \Delta_{w_i} \cW_{g+(2k+\ell-r-1)/2,\ell+2k,n} (w'_{\llbracket \ell\rrbracket},w_{\llbracket 2k\rrbracket} \mid z_{\llbracket n\rrbracket}) 
	\\ \notag &
	-\left(\sum_{2k+\ell = r+1} \frac{k}{\ell! (2k)!} \left(S_{z_0} W_{0,1}(X(z_0))\right)^\ell \left(\Delta_{z_0} W_{0,1}(X(z_0))\right)^{2k-2}\right) QLE_{g,n+1}
	\end{align}
	is holomorphic at $ z_0 \to p$.
\end{lemma}

Lemma~\ref{lem:secondRSpinLemma} and Lemma~\ref{lem:holomorphicexpressions} have the following direct corollary. Assume that the quadratic loop equations hold for all $(g',n')$ with $2g'-2+(n'+1)<2g-2+(n+1)$.

\begin{corollary} \label{cor:holomorphic} Depending on the ring \eqref{eq:R1def}-\eqref{eq:R3def} where $p$ belongs, one of the following expressions is holomorphic in $z_0$ near $p$: 
\begin{align}\label{eq:knownholo}
\left(\sum_{\substack{l\geq 0\\2N\geq 2}} \sum_{m=1}^\infty \dfrac{A^m-A^{-m}}{m}X^m\left(\dfrac{P}{Q}\,m\right)^{2N+l} \dfrac{(Sy)^l}{2^l\; l!}\dfrac{(\Delta y)^{2N-2}}{2^{2N-2}\;(2N)!}\; N\right)\; 	QLE_{g,n+1}, 
\\ \notag
\text{if } p\in R_1; \\ \label{eq:knownholo-R2}
\left(\sum_{\substack{l\geq 0\\2N\geq 2}} \sum_{m=1}^\infty \dfrac{A^{-m}}{m}\left((-1)^{2N+l}X^{-m}-X^m\right)\left(\dfrac{P}{Q}\,m\right)^{2N+l}  \dfrac{(Sy)^l}{2^l\; l!}\dfrac{(\Delta y)^{2N-2}}{2^{2N-2}\;(2N)!}\; N\right)\; 	QLE_{g,n+1}, 
\\ \notag
\text{if } p\in R_2; \\ \label{eq:knownholo-R3}
\left(\sum_{\substack{l\geq 0\\2N\geq 2}} \sum_{m=1}^\infty \dfrac{A^{-m}-A^{m}}{m}X^{-m}\left(-\dfrac{P}{Q}\,m\right)^{2N+l} \dfrac{(Sy)^l}{2^l\; l!}\dfrac{(\Delta y)^{2N-2}}{2^{2N-2}\;(2N)!}\; N\right)\; QLE_{g,n+1}, 
\\ \notag
\text{if } p\in R_3.
\end{align}
In order to shorten the notation we use here $X$ for $X(z_0)$ and $y$ for $y(z_0)$. 
\end{corollary}
\begin{proof} The proof repeats \emph{mutatis mutandis} the proof of the holomorphicity of~\cite[Equation~(42)]{KPS}. Indeed,  using Equation~\eqref{eq:SonDiagonal} and the fact that $S_{z_0} X(z_0) = 2 X(z_0)$ and $\Delta_{z_0} X(z_0) = 0$, we have:
\begin{align} \label{eq:ApplyS}
&  \cS_{z_0} \left[\pm \sum_{\substack{m\geq 1\\a\geq 0}} 
	\sum_{\substack{t \geq 1 \\ d\geq 0 }} A^{\pm m} B_a m^{2a-1} 
	\frac{(\pm m \frac PQ)^{t+2d} }{t!} \tilde Q_{d,\pm m,t}(z_0) \cW_{g-d-a,t,n} (w_{\llbracket t \rrbracket} \mid z_{\llbracket n \rrbracket}) \right] =
\\ \notag
& \pm  \sum_{\substack{m\geq 1\\a\geq 0}} A^{\pm m} B_a m^{2a-1} \sum_{d'\geq 0} \left(\pm m \frac PQ\right)^{2d'} [u^{2d'}]  \frac{u}{\zeta(u)} \frac{\zeta(u D_{X(z_0)})}{u D_{X(z_0)}} X(z_0)^{\pm m} \sum_{r\geq 0} \left(\pm m \frac PQ\right)^{r+1} \frac{1}{2^{r}}\times
\\ \notag
& 
\sum_{t=1}^{r+1}\frac{1}{t!} \sum_{\substack{2\alpha_1 + \dotsb + 2\alpha_t \\ + t = r+1}} \prod_{j=1}^t \bigg( \left[\vert_{w_j = z_0}\right] \frac{D_{X(w_j)}^{2\alpha_j}}{(2\alpha_j+1)!} \bigg) 
\sum_{\substack{I\sqcup J = \llbracket t \rrbracket \\ |I|\in 2\mathbb{Z} }} \prod_{i\in I} \Delta_{w_i} \prod_{j\in J} \cS_{w_j} \cW_{g-a+(t-r-1)/2,t,k} (w_{\llbracket t\rrbracket} \mid z_{\llbracket k\rrbracket}) \,.
\end{align}
According to Lemma~\ref{lem:secondRSpinLemma}, the last line in this expression is always a holomorphic function unless $a=d'=0$. In the latter case, we can subtract and add the term in the second line of Expression~\eqref{eq:expr-r}, that is,
\begin{align} \label{eq:QLECorrection}
&
\pm \sum_{\substack{m\geq 1}} \frac{A^{\pm m}} m X(z_0)^{\pm m} \sum_{r\geq 0} \left(\pm m \frac PQ\right)^{r+1} \frac{1}{2^{r}} \times
\\ \notag
& 
\left(\sum_{2k+\ell = r+1} \frac{k}{\ell! (2k)!} \left(S_{z_0} W_{0,1}(X(z_0))\right)^\ell \left(\Delta_{z_0} W_{0,1}(X(z_0))\right)^{2k-2}\right) QLE_{g,n+1}
\end{align}
We see that the difference of the right hand side of Equation~\eqref{eq:ApplyS} and Expression~\eqref{eq:QLECorrection} is an infinite series of functions holomorphic in $z_0$ at $z_0\to p$. Depending on the ring where $p$ belongs, we choose the $\pm$ signs exactly as in the corresponding two summands in the statement of Lemma~\ref{lem:holomorphicexpressions}. Then this infinite series of functions converges to a holomorphic function in the neighborhood of $p$ (the convergence is easy to see by exactly the same computation as in Sections~\ref{sec:ChangeFromQW-to-CheckH} and~\ref{sec:AnalyticCont}).

Thus, up to a factor of $1/2$,  Expression~\eqref{eq:BasicSymmetricExpression-R1} (respectively,~\eqref{eq:BasicSymmetricExpression-R2},~\eqref{eq:BasicSymmetricExpression-R3}) is equal in a neighborhood of the point $p$  to the sum of a holomorphic function and Expression~\eqref{eq:knownholo} (respectively,~\eqref{eq:knownholo-R2},~\eqref{eq:knownholo-R3}). Then the holomorphicity of Expressions~\eqref{eq:BasicSymmetricExpression-R1}-\eqref{eq:BasicSymmetricExpression-R3} implies the holomorphicity of Expressions~\eqref{eq:knownholo}-\eqref{eq:knownholo-R3}.
\end{proof}

\subsection{Analysis of the coefficients} The expressions in large brackets in \eqref{eq:knownholo}-\eqref{eq:knownholo-R3} can be further simplified. Recall that $X$ denotes $X(z_0)$ and $y$ denotes $y(z_0)$. We have:

\begin{lemma}\label{lem:coeffsimplify} The prefactors of ${QLE}_{g,n+1}$ 
in expressions~\eqref{eq:knownholo}-\eqref{eq:knownholo-R3} are equal to
\begin{align} \label{eq:FactorR1}
\dfrac{1}{2 \Delta y}\dfrac{P}{Q} \; \left(\dfrac{Az_0}{1-Az_0}-\dfrac{Az_0e^{-\frac{P}{Q}\Delta y}}{1-Az_0e^{-\frac{P}{Q}\Delta y}}\right) - \dfrac{1}{2 \Delta y}\dfrac{P}{Q} \; \left(\dfrac{A^{-1}z_0}{1-A^{-1}z_0}-\dfrac{A^{-1}z_0e^{-\frac{P}{Q}\Delta y}}{1-A^{-1}z_0e^{-\frac{P}{Q}\Delta y}}\right),
\\ 
\dfrac{1}{2 \Delta y}\dfrac{P}{Q} \; \left(\dfrac{A^{-1}z_0^{-1}e^{\frac{P}{Q}\Delta y}}{1-A^{-1}z_0^{-1}e^{\frac{P}{Q}\Delta y}}-\dfrac{A^{-1}z_0^{-1}}{1-A^{-1}z_0^{-1}}\right)
-\dfrac{1}{2 \Delta y}\dfrac{P}{Q} \; \left(\dfrac{A^{-1}z_0}{1-A^{-1}z_0}-\dfrac{A^{-1}z_0e^{-\frac{P}{Q}\Delta y}}{1-A^{-1}z_0e^{-\frac{P}{Q}\Delta y}}\right),
\\  \label{eq:FactorR3}
\dfrac{1}{2 \Delta y}\dfrac{P}{Q} \; \left(\dfrac{A^{-1}z_0^{-1}e^{\frac{P}{Q}\Delta y}}{1-A^{-1}z_0^{-1}e^{\frac{P}{Q}\Delta y}}-\dfrac{A^{-1}z_0^{-1}}{1-A^{-1}z_0^{-1}}\right)
-\dfrac{1}{2 \Delta y}\dfrac{P}{Q} \; \left(\dfrac{Az_0^{-1}e^{\frac{P}{Q}\Delta y}}{1-Az_0^{-1}e^{\frac{P}{Q}\Delta y}}-\dfrac{Az_0^{-1}}{1-Az_0^{-1}}\right),
\end{align}
respectively.
\end{lemma}

\begin{proof} It is convenient to change the summation index from $N$ to $N'=N-1$. Note that
\begin{align}
&\sum_{l,N'=0}^\infty\; \sum_{m=1}^\infty \dfrac{A^mX^m}{m}\left(\dfrac{P}{Q}\,m\right)^{2N'+l+2} \dfrac{(Sy)^l}{2^l\; l!}\dfrac{(\Delta y)^{2N'}}{2^{2N'}\;(2N')!}\; \dfrac{1}{2\,(2N'+1)}
\\ \nonumber
&=\dfrac{1}{2}\left(\dfrac{P}{Q}\right)^2\sum_{m=1}^\infty m^2 \dfrac{A^mX^m}{m}\sum_{l=0}^\infty \dfrac{1}{l!} \left(\dfrac{1}{2}\dfrac{P}{Q}\, m\, Sy\right)^l \sum_{N'=0}^\infty \dfrac{1}{(2N'+1)!}\left(\dfrac{1}{2}\dfrac{P}{Q}\, m\, \Delta y \right)^{2N'} \\ \nonumber
&=\dfrac{1}{2}\left(\dfrac{P}{Q}\right)^2\sum_{m=1}^\infty m^2 \dfrac{A^mX^m}{m} \exp\left(\dfrac{1}{2}\dfrac{P}{Q}\, m\, Sy\right) \dfrac{\exp\left(\dfrac{1}{2}\dfrac{P}{Q}\, m\, \Delta y\right)-\exp\left(-\dfrac{1}{2}\dfrac{P}{Q}\, m\, \Delta y\right)}{\dfrac{P}{Q}\, m\, \Delta y}\\ \nonumber
&=\dfrac{1}{2 \Delta y}\dfrac{P}{Q} \; \sum_{m=1}^\infty\left(\left(AXe^{\frac{1}{2}\frac{P}{Q}Sy}e^{\frac{1}{2}\frac{P}{Q}\Delta y}\right)^m-\left(AXe^{\frac{1}{2}\frac{P}{Q}Sy}e^{-\frac{1}{2}\frac{P}{Q}\Delta y}\right)^m\right)\\ \nonumber
&=\dfrac{1}{2 \Delta y}\dfrac{P}{Q} \; \left(\dfrac{AXe^{\frac{P}{Q}\left(\frac{Sy}{2}+\frac{\Delta y}{2}\right)}}{1-AXe^{\frac{P}{Q}\left(\frac{Sy}{2}+\frac{\Delta y}{2}\right)}}-\dfrac{AXe^{\frac{P}{Q}\left(\frac{Sy}{2}-\frac{\Delta y}{2}\right)}}{1-AXe^{\frac{P}{Q}\left(\frac{Sy}{2}-\frac{\Delta y}{2}\right)}}\right)\\ \nonumber
&=\dfrac{1}{2 \Delta y}\dfrac{P}{Q} \; \left(\dfrac{Az_0}{1-Az_0}-\dfrac{Az_0e^{-\frac{P}{Q}\Delta y}}{1-Az_0e^{-\frac{P}{Q}\Delta y}}\right)\,.
\end{align}
This computation as well as its close analogs for 
\begin{align}
&\sum_{l,N'=0}^\infty\; \sum_{m=1}^\infty \dfrac{A^{\pm m}X^{\pm m}}{m}\left(\pm \dfrac{P}{Q}\,m\right)^{2N'+l+2} \dfrac{(Sy)^l}{2^l\; l!}\dfrac{(\Delta y)^{2N'}}{2^{2N'}\;(2N')!}\; \dfrac{1}{2\,(2N'+1)}
\end{align}
imply the statement of the lemma.
\end{proof}

Assume that $p\not= 0,1,-1, A, A^{-1}$. It is a condition on $A$ that we allow according to Remark~\ref{rem:Agen}. Under this assumption we have:
\begin{lemma}\label{lem:coeffinvert} Expressions \eqref{eq:FactorR1}-\eqref{eq:FactorR3} give invertible holomorphic functions in the variable $z$ near $z=p$.
\end{lemma}

\begin{proof} At the branch point $z\to p$ we have $\Delta y \rightarrow 0$. The corresponding limits of ~\eqref{eq:FactorR1}-\eqref{eq:FactorR3} exist. Moreover, they all coincide and are equal to 
\begin{align}
&\frac{1}{2}\dfrac{P^2}{Q^2} \dfrac{A p}{(Ap-1)^2} - \frac{1}{2}\dfrac{P^2}{Q^2} \dfrac{A^{-1} p}{(A^{-1}p-1)^2} = \frac{1}{2}\dfrac{P^2 p (p^2-1) (A^{-1}-A) }{Q^2 (p-A)^2(p-A^{-1})^2}\,.
\end{align}
Under the assumption on $p$ that we made these limits are not equal to zero.
\end{proof}

\subsection{Proof of the quadratic loop equation} \label{sec:ProofQLE} In this Section we finally prove Theorem~\ref{thm:QLE-equiv} that the quadratic loop equations hold in our case, that is, we prove that  $QLE_{g,n+1}$  is holomorphic in $z_0$ at $z_0\to p$. The proof follows the same scheme as in~\cite{DKPS} and~\cite{KPS}, we use induction on the negative Euler characteristic $2g-2+(n+1)$.
\begin{proof}[Proof of Theorem~\ref{thm:QLE-equiv}]
The base case $g=n=0$ is evident since the function $y$ is holomorphic near the critical points of $dX$. Under the assumption that the quadratic loop equations are valid for all $2g'-2+(n'+1)<2g-2+(n+1)$, Corollary~\ref{cor:holomorphic} implies that $QLE_{g,n+1}$  is holomorphic in $z_0$ at $z_0\to p$ 
when multiplied 
by one of the prefactors~\eqref{eq:FactorR1}-\eqref{eq:FactorR3}, depending on the ring where the point $p$ belongs. According to Lemmata~\ref{lem:coeffsimplify} and~\ref{lem:coeffinvert} these prefactors are invertible holomorphic functions in $z_0$ at $z_0\to p$. 
Therefore, $QLE_{g,n+1}$ is holomorphic as well. 
\end{proof}


\bibliographystyle{alphaurl}
\bibliography{homfly_recursion}
\end{document}